\pdfoutput=1
\documentclass[a4paper,UKenglish,cleveref, autoref, thm-restate,authorcolumns, colorlinks]{lipics-v2019}

\title{Weighted Tiling Systems for Graphs: Evaluation Complexity} 

\author{C. Aiswarya}{Chennai Mathematical Institute, India and IRL ReLaX, CNRS France}{aiswarya@cmi.ac.in}{https://orcid.org/0000-0002-4878-7581}
{Supported by DST Inspire}

\author{Paul Gastin}{LSV, ENS Paris-Saclay, CNRS, Universit\'e Paris-Saclay,
France}{paul.gastin@ens-paris-saclay.fr}{https://orcid.org/0000-0002-1313-7722}{}

\authorrunning{C. Aiswarya and P. Gastin} 

\Copyright{C. Aiswarya and Paul Gastin} 



\category{} 

\relatedversion{An extended abstract of preliminary version is accepted at FSTTCS'20.}

\supplement{}

\funding{Supported by IRL ReLaX}

\nolinenumbers 

\hideLIPIcs  

\usepackage[pdflatex,recompilepics=false]{gastex}
\usepackage{MnSymbol}
\usepackage{stmaryrd, wasysym}
\usepackage{newfloat}
\DeclareFloatingEnvironment[name=Problem]{problem}
\usepackage[inline]{enumitem}
\usepackage{xspace}
\usepackage{algorithm}
\usepackage{algpseudocode}
\algblockdefx[MATCH]{Match}{EndMatch}%
  [1]{\textbf{match} #1 \textbf{with}}%
  {\textbf{end match}}
\algdef{SxN}[CASE]{Case}{EndCase}%
  [1]{{#1}:}
\algdef{SxN}[IFTHEN]{IfThen}{EndIfThen}%
  [2]{\textbf{if}~{#1}~\textbf{then} {#2}~\textbf{end if}}
\algdef{SxN}[FORONELINE]{ForOne}{EndForOne}%
  [2]{\textbf{for}~{#1}~\textbf{do} {#2~\textbf{end for}}}
\algdef{SxN}[FORALLONELINE]{ForAllOne}{EndForAllOne}%
  [2]{\textbf{for all}~{#1}~\textbf{do} {#2~\textbf{end for}}}

\usepackage{todonotes}

\newcommand{\Ptime}{\textrm{P}}
\newcommand{\NP}{\textrm{NP}}

\newcommand{\FP}{\textrm{FP}}
\newcommand{\sharpP}{\#\textrm{P}}
\newcommand{\gapP}{\textrm{GapP}}
\newcommand{\FPSPACE}{\textsc{FPSpace}}

\newcommand{\FPtoNPlog}{\FP^{\NP[log]}}

\newcommand{\M}{M}

\newcommand{\feval}{\textsf{FVal}}

\newcommand{\fval}{\textsf{Eval}}

\newcommand{\N}{\mathbb N}
\newcommand{\Z}{\mathbb Z}
\newcommand{\Q}{\mathbb Q}

\newcommand{\Boolean}{\textsf{Boolean}}
\newcommand{\Natural}{\textsf{Natural}}
\newcommand{\Integer}{\textsf{Integer}}
\newcommand{\Rational}{\textsf{Rational}}
\newcommand{\Rationalp}{\textsf{Rational}^+}

\newcommand{\minplusN}{\textsf{min-plus-}\N}
\newcommand{\maxplusN}{\textsf{max-plus-}\N}
\newcommand{\minplusZ}{\textsf{min-plus-}\Z}
\newcommand{\maxplusZ}{\textsf{max-plus-}\Z}

\newcommand{\SR}{\mathbb S}
\newcommand{\szero}{0_\SR}
\newcommand{\sone}{1_\SR}
\newcommand{\stimes}{\otimes}
\newcommand{\splus}{\oplus}

\newcommand{\Edgenames}{\Gamma}
\newcommand{\Nodelabels}{\Sigma}
\newcommand{\Ename}{\gamma}
\newcommand{\Nlabeling}{\lambda}
\newcommand{\Edgesin}{{\Edgenames_\text{in}}}
\newcommand{\Edgesout}{{\Edgenames_\text{out}}}
\newcommand{\Types}{\textsf{Types}}
\newcommand{\type}{\tau}
\newcommand{\typemap}{\textsf{type}}
\newcommand{\A}{\mathcal A}
\newcommand{\States}{Q}
\newcommand{\Tiles}{\Delta}
\newcommand{\Weight}{\textsf{wgt}}
\newcommand{\run}{\rho}
\newcommand{\fin}{f_\text{in}}
\newcommand{\fout}{f_\text{out}}
\newcommand{\tile }{\textsf{tile}}

\newcommand\sem[1]{[\![ #1 ]\!]}
\newcommand\wsem[1]{\{\!{\mid} #1 {\mid}\!\}}

\newcommand{\fempty}{f_\emptyset}

\newcommand{\eright}{\rightarrow}
\newcommand{\edown}{\downarrow}

\newcommand{\qthis}{\bigcirc}
\newcommand{\qur}{\curvearrownwse}
\newcommand{\qdr}{\curvearrownesw}
\newcommand{\qdl}{\curvearrowsenw}
\newcommand{\qul}{\curvearrowswne}

\newcommand{\qcol}{\rotatebox[origin=c]{90}{$\boxminus$}}
\newcommand{\qrow}{\boxminus}
\newcommand{\qcr}{\boxplus}
\newcommand{\qnone}{\raisebox{.5pt}{\scalebox{1.15}{$\boxempty$}}}

\newcommand{\wts}{\mathcal T}
\newcommand{\ma}{\mathcal{M}}

\begin{document}

\maketitle

\begin{abstract}
  We consider weighted tiling systems to represent functions from graphs to a
  commutative semiring such as the Natural semiring or the Tropical semiring.  The
  system labels the nodes of a graph by its states, and checks if the
  neighbourhood of every node belongs to a set of permissible tiles, and
  assigns a weight accordingly.  The weight of a labeling is the
  semiring-product of the weights assigned to the nodes, and the weight of the
  graph is the semiring-sum of the weights of labelings.  We show that we can
  model interesting algorithmic questions using this formalism - like computing
  the clique number of a graph or computing the permanent of a matrix.  The
  evaluation problem is, given a weighted tiling system and a graph, to compute
  the weight of the graph.  We study the complexity of the evaluation problem
  and give tight upper and lower bounds for several commutative semirings.
  Further we provide an efficient evaluation algorithm if the input graph is of
  bounded tree-width.
\end{abstract}
%
\begin{gpicture}[name=gpic:permanent-matrix, ignore]
\gasset{Nw=6,Nh=6,AHLength=1.8,AHlength=1, Nframe = n, Nadjust=n, Nframe=n,fillcolor=white }
\unitlength=0.9mm

\node(n00)(0,0){0}
\node(n01)(0,10){0}
\node(n02)(0,20){1}
\node(n03)(0,30){0}
\node(n04)(0,40){1}
\node(n10)(-10,0){1}
\node(n11)(-10,10){0}
\node(n12)(-10,20){1}
\node(n13)(-10,30){0}
\node(n14)(-10,40){1}
\node(n20)(-20,0){0}
\node(n21)(-20,10){1}
\node(n22)(-20,20){1}
\node(n23)(-20,30){0}
\node(n24)(-20,40){1}
\node(n30)(-30,0){0}
\node(n31)(-30,10){1}
\node(n32)(-30,20){0}
\node(n33)(-30,30){1}
\node(n34)(-30,40){1}
\node(n40)(-40,0){0}
\node(n41)(-40,10){0}
\node(n42)(-40,20){0}
\node(n43)(-40,30){0}
\node(n44)(-40,40){1}

\drawedge(n04,n03){}
\drawedge(n03,n02){}
\drawedge(n02,n01){}
\drawedge(n01,n00){}
\drawedge(n14,n13){}
\drawedge(n13,n12){}
\drawedge(n12,n11){}
\drawedge(n11,n10){}
\drawedge(n24,n23){}
\drawedge(n23,n22){}
\drawedge(n22,n21){}
\drawedge(n21,n20){}
\drawedge(n34,n33){}
\drawedge(n33,n32){}
\drawedge(n32,n31){}
\drawedge(n31,n30){}
\drawedge(n44,n43){}
\drawedge(n43,n42){}
\drawedge(n42,n41){}
\drawedge(n41,n40){}

\drawedge(n40,n30){}
\drawedge(n30,n20){}
\drawedge(n20,n10){}
\drawedge(n10,n00){}
\drawedge(n41,n31){}
\drawedge(n31,n21){}
\drawedge(n21,n11){}
\drawedge(n11,n01){}

\drawedge(n42,n32){}
\drawedge(n32,n22){}
\drawedge(n22,n12){}
\drawedge(n12,n02){}

\drawedge(n43,n33){}
\drawedge(n33,n23){}
\drawedge(n23,n13){}
\drawedge(n13,n03){}

\drawedge(n44,n34){}
\drawedge(n34,n24){}
\drawedge(n24,n14){}
\drawedge(n14,n04){}
\end{gpicture}
\begin{gpicture}[name=gpic:permanent-run, ignore]
\gasset{Nw=6,Nh=6,AHLength=1.8,AHlength=1, Nframe = n, Nadjust=n, Nframe=n,fillcolor=white }
\unitlength=0.9mm
\node(m00)(0,0){\huge$\color{red}\qul$}
\node(m01)(0,10){\huge\color{red}$\qul$}
\node(m02)(0,20){\huge\color{red}$\bigcirc$}
\node(m03)(0,30){\huge\color{red}$\qdl$}
\node(m04)(0,40){\huge\color{red}$\qdl$}
\node(m10)(-10,0){\huge$\color{red}\qul$}
\node(m11)(-10,10){\huge$\color{red}\qul$}
\node(m12)(-10,20){\huge$\color{red}\qur$}
\node(m13)(-10,30){\huge\color{red}$\bigcirc$}
\node(m14)(-10,40){\huge\color{red}$\qdl$}
\node(m20)(-20,0){\huge$\color{red}\qul$}
\node(m21)(-20,10){\huge\color{red}$\bigcirc$}
\node(m22)(-20,20){\huge$\color{red}\qdr$}
\node(m23)(-20,30){\huge$\color{red}\qdr$}
\node(m24)(-20,40){\huge$\color{red}\qdl$}
\node(m30)(-30,0){\huge$\color{red}\qul$}
\node(m31)(-30,10){\huge$\color{red}\qur$}
\node(m32)(-30,20){\huge$\color{red}\qur$}
\node(m33)(-30,30){\huge$\color{red}\qur$}
\node(m34)(-30,40){\huge\color{red}$\bigcirc$}
\node(m40)(-40,0){\huge\color{red}$\bigcirc$}
\node(m41)(-40,10){\huge$\color{red}\qdr$}
\node(m42)(-40,20){\huge$\color{red}\qdr$}
\node(m43)(-40,30){\huge$\color{red}\qdr$}
\node(m44)(-40,40){\huge$\color{red}\qdr$}
\drawedge(m04,m03){}
\drawedge(m03,m02){}
\drawedge(m02,m01){}
\drawedge(m01,m00){}
\drawedge(m14,m13){}
\drawedge(m13,m12){}
\drawedge(m12,m11){}
\drawedge(m11,m10){}
\drawedge(m24,m23){}
\drawedge(m23,m22){}
\drawedge(m22,m21){}
\drawedge(m21,m20){}
\drawedge(m34,m33){}
\drawedge(m33,m32){}
\drawedge(m32,m31){}
\drawedge(m31,m30){}
\drawedge(m44,m43){}
\drawedge(m43,m42){}
\drawedge(m42,m41){}
\drawedge(m41,m40){}

\drawedge(m40,m30){}
\drawedge(m30,m20){}
\drawedge(m20,m10){}
\drawedge(m10,m00){}
\drawedge(m41,m31){}
\drawedge(m31,m21){}
\drawedge(m21,m11){}
\drawedge(m11,m01){}

\drawedge(m42,m32){}
\drawedge(m32,m22){}
\drawedge(m22,m12){}
\drawedge(m12,m02){}

\drawedge(m43,m33){}
\drawedge(m33,m23){}
\drawedge(m23,m13){}
\drawedge(m13,m03){}

\drawedge(m44,m34){}
\drawedge(m34,m24){}
\drawedge(m24,m14){}
\drawedge(m14,m04){}

\gasset{Nw=3,Nh=3,AHLength=1.8,AHlength=1, Nframe = n, Nadjust=n, Nframe=n,fillcolor=white }
\unitlength=0.9mm

\node(n00)(0,0){0}
\node(n01)(0,10){0}
\node(n02)(0,20){1}
\node(n03)(0,30){0}
\node(n04)(0,40){1}
\node(n10)(-10,0){1}
\node(n11)(-10,10){0}
\node(n12)(-10,20){1}
\node(n13)(-10,30){0}
\node(n14)(-10,40){1}
\node(n20)(-20,0){0}
\node(n21)(-20,10){1}
\node(n22)(-20,20){1}
\node(n23)(-20,30){0}
\node(n24)(-20,40){1}
\node(n30)(-30,0){0}
\node(n31)(-30,10){1}
\node(n32)(-30,20){0}
\node(n33)(-30,30){1}
\node(n34)(-30,40){1}
\node(n40)(-40,0){0}
\node(n41)(-40,10){0}
\node(n42)(-40,20){0}
\node(n43)(-40,30){0}
\node(n44)(-40,40){1}
\end{gpicture}
\begin{gpicture}[name=gpic:clique-graph,ignore]
\gasset{Nh=5, Nw =5, Nframe = y}
\node(A)(10,0){A}
\node(B)(30, 0){B}
\node(C)(38,18){C}
\node(D)(20,30){D}
\node(E)(2,18){E}
\gasset{AHnb = 0}
\drawedge(A,B){}
\drawedge(B,C){}
\drawedge(A,D){}
\drawedge(B,E){}
\drawedge(A,E){}
\drawedge(D,C){}
\drawedge(C,E){}
\drawedge(B,D){}
\end{gpicture}
\begin{gpicture}[name=gpic:clique-matrix,ignore]
\gasset{Nw=6,Nh=6,AHLength=1.8,AHlength=1, Nframe = n}
\node(AE)(0,0){1}
\node(BE)(0,10){1}
\node(CE)(0,20){1}
\node(DE)(0,30){0}
\node(E)(0,40){1}
\node(AD)(-10,0){1}
\node(BD)(-10,10){1}
\node(CD)(-10,20){1}
\node(D)(-10,30){1}
\node(AC)(-20,0){0}
\node(BC)(-20,10){1}
\node(C)(-20,20){1}
\node(AB)(-30,0){1}
\node(B)(-30,10){1}
\node(A)(-40,0){1}

\drawedge(A,AB){}
\drawedge(AB,AC){}
\drawedge(AC,AD){}
\drawedge(AD,AE){}
\drawedge(B,BC){}
\drawedge(BC,BD){}
\drawedge(BD,BE){}
\drawedge(C,CD){}
\drawedge(CD,CE){}
\drawedge(D,DE){}

\drawedge(E,DE){}
\drawedge(DE,CE){}
\drawedge(CE,BE){}
\drawedge(BE,AE){}
\drawedge(D,CD){}
\drawedge(CD,BD){}
\drawedge(BD,AD){}
\drawedge(C,BC){}
\drawedge(BC,AC){}
\drawedge(B,AB){}

\gasset{ExtNL=y,NLdist=1, NLangle=135}
\node(a)(-40,0){A}
\node(b)(-30,10){B}
\node(c)(-20,20){C}
\node(d)(-10,30){D}
\node(e)(0,40){E}
\end{gpicture}
\begin{gpicture}[name=gpic:clique-run, ignore]
\gasset{Nw=6,Nh=6,AHLength=1.8,AHlength=1, Nframe = n}
\node(AE)(0,0){\huge$\color{red}\qcol$}
\node(BE)(0,10){\huge$\color{red}\qcr$}
\node(CE)(0,20){\huge$\color{red}\qcr$}
\node(DE)(0,30){\huge$\color{red}\qcol$}
\node(E)(0,40){\huge$\color{red}\qcr$}
\node(AD)(-10,0){\huge$\color{red}\qnone$}
\node(BD)(-10,10){\huge$\color{red}\qrow$}
\node(CD)(-10,20){\huge$\color{red}\qrow$}
\node(D)(-10,30){\huge$\color{red}\qnone$}
\node(AC)(-20,0){\huge$\color{red}\qcol$}
\node(BC)(-20,10){\huge$\color{red}\qcr$}
\node(C)(-20,20){\huge$\color{red}\qcr$}
\node(AB)(-30,0){\huge$\color{red}\qcol$}
\node(B)(-30,10){\huge$\color{red}\qcr$}
\node(A)(-40,0){\huge$\color{red}\qnone$}

\node(AE)(0,0){1}
\node(BE)(0,10){1}
\node(CE)(0,20){1}
\node(DE)(0,30){0}
\node(E)(0,40){1}
\node(AD)(-10,0){1}
\node(BD)(-10,10){1}
\node(CD)(-10,20){1}
\node(D)(-10,30){1}
\node(AC)(-20,0){0}
\node(BC)(-20,10){1}
\node(C)(-20,20){1}
\node(AB)(-30,0){1}
\node(B)(-30,10){1}
\node(A)(-40,0){1}

\drawedge(A,AB){}
\drawedge(AB,AC){}
\drawedge(AC,AD){}
\drawedge(AD,AE){}
\drawedge(B,BC){}
\drawedge(BC,BD){}
\drawedge(BD,BE){}
\drawedge(C,CD){}
\drawedge(CD,CE){}
\drawedge(D,DE){}

\drawedge(E,DE){}
\drawedge(DE,CE){}
\drawedge(CE,BE){}
\drawedge(BE,AE){}
\drawedge(D,CD){}
\drawedge(CD,BD){}
\drawedge(BD,AD){}
\drawedge(C,BC){}
\drawedge(BC,AC){}
\drawedge(B,AB){}
\end{gpicture}
%

\section{Introduction}
\label{sec:introduction}

Weighted automata have been classically studied over words, as they naturally
extend automata from representing languages to representing functions from words
to a semiring.

We are interested in finite state formalisms for representing functions from
\emph{graphs} to a semiring.  Many natural algorithmic questions on graphs are
about computing a function, such as the clique number, weight of the shortest
path etc.  It is interesting to see if one can design weighted automata to model
such problems.  Further can one design efficient algorithms for problems modeled
by such weighted automata?

We study weighted tiling systems (WTS),
 a variant of the weighted graph
automata of Droste and D\"uck \cite{DrosteD15}, motivated by the graph
acceptors of Thomas~\cite{Thomas91}.  This subsumes many quantitative models
that have been studied on words, trees~\cite{DrostePV05,DrosteV06}, nested
words~\cite{Mathissen_2010}, pictures~\cite{Fichtner11}, Mazurkiewicz
traces~\cite{DrosteG99,Meinecke06,BolligM07}, etc. The reader is referred to 
the handbook \cite{DrosteKV2009handbook} for more details and references. 
Many of these works are mainly interested in expressivity
questions, and show that the model has good expressive power.  The model is
also easy to understand as it is formulated in terms of tiling/colouring
respecting local constraints.  We reiterate the expressivity by modeling
computational problems on graphs using this model.  Our focus is on the
computational complexity of the evaluation problem. It is closer in spirit to
\cite{Gastin_2014} which provides an efficient evaluation algorithm for 
weighted pebble automata on words.

We show that many algorithmic questions,
like computing the clique number, computing the permanent of a matrix, or
counting variants of SAT, can be naturally modeled using this formalism.  We
investigate the computational complexity of the evaluation problem
and obtain tight upper- and lower-bounds for various
semirings.  

To give more details, a WTS has a finite number of states and a run labels the
vertices of a graph with states.  The tiles (analogous to transitions) observe
the neighbourhood of a vertex under the labeling, and assign a weight
accordingly.  The weight of the run is the semiring-product of the weights thus
assigned, and the weight assigned to a graph is the semiring-sum of the weights
of the runs.  We only consider commutative semirings and hence the order in
which the product is taken does not matter.

The evaluation problem is to compute the weight of an input graph in an input
WTS. We study the computational complexity of this problem for various
semirings.  Over Natural semiring and non-negative rationals, the problem is
shown to be $\sharpP$-complete.  Over integers and rationals the problem is
$\gapP$-complete.  Over tropical semirings -- $(\N, \max, +), (\Z, \max, +),
(\N, \min, +), (\Z, \min, +)$ -- the problem is $\FPtoNPlog$ complete.

We further consider the evaluation problem for graphs of bounded tree-width and
show that they are computable in time polynomial in the WTS and linear in the
graph.  Bounded tree-width captures a variety of formal models of concurrent and
infinite state systems such as Mazurkiewicz traces, nested words, and decidable
under-approximations of message passing automata or multi-pushdown automata
\cite{MadhusudanP11, CGN12, AGN-atva14}.

Even though our focus is evaluation, and not expressiveness of the model, we get
a deep insight into the modeling power of this formalism through the upper and
lower complexity bounds.  For instance, we cannot polynomially encode the
traveling salesman problem (lower bound $\textsf{FP}^\NP$) in our formalism over
tropical semiring (upper bound $\textsf{FP}^{\NP[\log]}$) unless the polynomial
hierarchy collapses \cite{Krentel88}.

\section{Model}
\label{sec:model}

First we will fix the notations for semirings, graphs and then introduce the WTS
formally.

\subparagraph{Preliminaries.}

Let $\N$ denote the set of natural numbers including $0$, $\Z$ the integers, and
$\Q$ the rationals.

\smallskip

Let $A=\{a_1, \dots a_n\}$ and $B$ be two sets.  
We sometimes write a function $f \colon A \to B$ explicitly by listing the
image of each element: $f = [a_1 \mapsto
f(a_1), \dots, a_n \mapsto f(a_n)]$.  The set of all functions from $A$ to $B$ is
denoted $B^A$.  If $A$ is $\emptyset$ then the only relation (and hence
function) from $A$ to $B$ is $\emptyset$.  We denote this trivial empty function
by $\fempty$.

\smallskip

Let $\M$ be a non-deterministic Turing machine.  The number of accepting runs
of $\M$ on an input $x$ is denoted $\#\M(x)$, and the number of rejecting runs
of $\M$ on $x$ is denoted $\#\overline{\M}(x)$.  

\smallskip

A semiring is an algebraic structure $\SR = (S, \splus, \stimes, \szero, \sone)
$ where $S$ is a set, $\oplus$ and $\otimes$ are two binary operations on $S$,
$(S, \oplus, \szero) $ is a commutative monoid, $(S,\otimes,\sone)$ is a monoid,
$\otimes$ distributes over $\oplus$, $\szero$ is an annihilator for $\otimes$.
A semiring is \emph{commutative} if $\stimes$ is commutative.

Examples include $\Boolean = (\{0, 1\}, \vee, \wedge, 0, 1)$, $\Natural = (\N,
+, \times, 0, 1)$, $\Integer = (\Z, +, \times, 0, 1)$, $\Rational = (\Q, +,
\times, 0, 1)$ and $\Rationalp = (\Q_{\ge 0}, +,
\times, 0, 1)$.  Further
examples are tropical semirings: 
$\maxplusN = (\N\cup\{-\infty\},\max,+,-\infty,0)$,
$\maxplusZ = (\Z\cup\{-\infty\},\max,+,-\infty,0)$, 
$\minplusN = (\N\cup\{+\infty\},\min,+,+\infty,0)$ and  
$\minplusZ = (\Z\cup\{+\infty\},\min,+,+\infty,0)$.
We will consider only these semirings in this paper.  Note that
all these semirings are commutative.

\subparagraph{Graphs.}
We consider graphs with different sorts of edges.  \ For example, a grid will
have horizontal successor edges, and vertical successor edges.  A binary tree
will have left-child relations and right-child relations.  Message sequence
charts will have process-successor relations and message send-receive relations.
These graphs have bounded degree, and for each sort of edge, a vertex will have
at most one outgoing/incoming edge of that sort\footnote{This choice is mainly
for notational convenience, and is not really a restriction provided we consider
only bounded degree graphs.  Another option would be to enumerate the neighbours
in some order and address a neighbour as the $i$th incoming/outgoing
neighbour.}.  Our definition of graphs below allows to capture such graph
classes.

\smallskip
Let $\Edgenames$ be a finite set of edge names, and let $\Nodelabels$ be a
finite set of node labels.  A $(\Edgenames,\Nodelabels)$-graph $G =
(V,(E_\Ename)_{\Ename\in \Edgenames},\Nlabeling)$ has a finite set of vertices
$V$, an edge relation $E_\Ename \subseteq V \times V$ for every $\Ename \in
\Edgenames$, and a mapping $\Nlabeling \colon V \to \Nodelabels$ assigning a label from
$\Nodelabels$ to each vertex $v \in V$.  The graphs we consider will have at
most one outgoing edge and at most one incoming edge for every edge name.  That
is, for each $\Ename \in \Edgenames$, for all $v \in V$, $|\{u \mid (v,u) \in
E_\Ename \}| \le 1$ and $|\{u \mid (u,v) \in E_\Ename \}| \le 1$.

The type of a vertex is determined by the set of names of incoming edges
and the set of names of outgoing edges.  
For example, the root of a tree has no incoming left-child or right-child edges
and leaves of a tree have no outgoing left- or right-child.  A type $\type =
(\Edgesin, \Edgesout) $ indicates that the set of incoming (resp.\ outgoing)
edge names is $\Edgesin$ (resp.\ $\Edgesout$).
Let $\Types = 2^\Edgenames \times 2^\Edgenames$ be the set of all
types.  We define $\typemap \colon V \to \Types$ and use $\typemap(v)$ to denote
the type of vertex $v$.

\begin{remark}\label{rem:arbitrarydegree}
  Even though we consider only bounded degree graphs, we are able to model graph
  functions on arbitrary graphs (even edge weighted) as illustrated in the
  examples below.  Basically an arbitrary graph is input via its adjacency
  matrix, which is naturally a grid, a special case of the graphs that we can
  handle.  We can even model problems on arbitrary graphs with edge weights.
\end{remark}

\newcommand{\occ}{\textsf{Occ}}

\subparagraph{A weighted Tiling System}
is a finite state mechanism for defining functions from a class of graphs to a
weight domain. It has a finite set
of states and a set of permissible tiles for
each type of vertices.
Formally, a \emph{weighted tiling system} (WTS) over
$(\Edgenames,\Nodelabels)$-graphs and a semiring $\SR = (S, \oplus, \otimes,
\szero, \sone) $ is a tuple $\wts = (\States, \Tiles,  \Weight)$ where
\begin{itemize}[nosep]
  \item  $\States$ is the finite set of states,
  
  \item $\Tiles = \bigcup_{\type \in \Types}\Tiles_\type$ --- 
  for a type $\type = (\Edgesin, \Edgesout) \in \Types$, the set $
  \Tiles_\type\subseteq \States^\Edgesin \times \States \times \Nodelabels \times
  \States^\Edgesout $ gives the set of permissible tiles of type $\type$,

  \item $\Weight \colon \Tiles \to S$,  assigns a  weight for each tile. 
\end{itemize}

\smallskip

A run $\run$ of $\wts$ on a graph $G = (V,(E_\Ename)_{\Ename\in
\Edgenames},\Nlabeling)$ is a labeling of the vertices by states that conforms
to $\Tiles$.  Given a labeling $\run\colon V \to \States$, for a vertex $v \in
V$ with $\typemap(v) = (\Edgesin, \Edgesout)$ we define the tile of $v$ wrt.\ 
$\run$ to be $\tile_\run(v)=(\fin, \run(v), \Nlabeling(v), \fout)$ where $\fin
\colon \Edgesin \to \States$ is given by $\Ename \mapsto \run(u)$ if $(u,v) \in
E_\Ename$ and $\fout \colon \Edgesout \to \States$ is given by $\Ename \mapsto
\run(u)$ if $(v,u) \in E_\Ename$.  
A labeling $\rho\colon V \to \States$ is a \emph{run} if for each $v \in V$,
$\tile_\run(v) \in \Tiles_{\typemap(v)}$.

\smallskip 

The \emph{weight of a run} $\run$, denoted $\Weight(\run)$, is the
product of the weights of the tiles in $\run$. 
With commutative semirings, we do not need to specify an order for this product. 
The value $\sem{\wts}(G)$ computed by $\wts$ for a graph $G$
is the sum of the weights of the runs.  That is,
\begin{align*}
 \sem\wts(G) &= \bigoplus_{\run \mid \run\text{ is a run of $\wts$ on  }G} \Weight(\run) 
 &\qquad\qquad
 \Weight(\run) &= \bigotimes_{v \in V} \Weight(\tile_\run(v)) .&
\end{align*}

\begin{remark}
  The WTS is a variant of the weighted graph automata (WGA) of \cite{DrosteD15}.
  There are two main differences.  First, WGA admits tiles of bigger radius and
  the tile size is a parameter.  This is not more powerful, as it can be
  realized with immediate neighborhood tiles like in WTS. Second, WGA allows
  occurrence constraints.  We discuss this in more detail in
  Section~\ref{sec:discussion}.
\end{remark}

We give some examples of WTS below, which will also serve as reductions proving
complexity lower-bounds in Section~\ref{sec:complexity-arbitrary-graphs}.

\begin{example}[A WTS to compute the clique number of a graph]
  \label{ex:clique-number} 
  The \emph{clique number} of a graph is the size of the largest clique in the graph. 

  The graphs on which we want to compute the clique number have unbounded
  degrees indeed.  In our setting we consider only bounded degree graphs.  Hence
  we need to encode any arbitrary graph as a bounded degree graph.  One way to
  do that is to consider the adjacency matrix and represent this matrix using a
  grid graph.

  For the particular case of clique number, our input is an undirected graph, so
  we will consider a lower-right triangular matrix in a lower-right triangular
  grid graph.  For this we let $\Edgenames = \{\eright, \edown\}$ and
  $\Nodelabels = \{0, 1\}$.  The labels of all diagonal vertices are $1$.  A
  graph is depicted in Figure~\ref{fig:clique-graph} and its lower-right
  triangular adjacency matrix is depicted in Figure~\ref{fig:clique-matrix}.

\medskip\noindent
\begin{minipage}[b]{3cm}
\gusepicture[scale=.8]{gpic:clique-graph}
\captionof{figure}{A graph}
\label{fig:clique-graph}
\end{minipage}\hfill
\begin{minipage}[b]{5cm}
\begin{center}
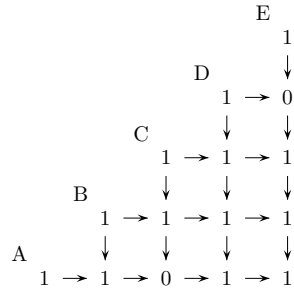

\gusepicture[scale=.8]{gpic:clique-matrix}
\end{center}\captionof{figure}{The lower-right triangular adjacency matrix of the graph of Figure~\ref{fig:clique-graph}  as a grid graph}
\label{fig:clique-matrix}
\end{minipage}\hfill
\begin{minipage}[b]{5cm}
\begin{center}
\gusepicture[scale=.8]{gpic:clique-run}
\end{center}
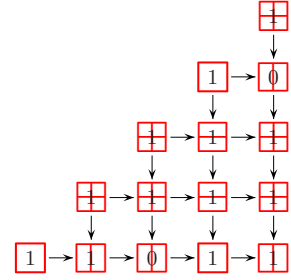
\captionof{figure}{A run. Three tiles B, C and E gets weights 1, and hence the weight of this run is 3.}
\label{fig:clique-run}
\end{minipage}

  \medskip
  We will now construct a WTS over the tropical semiring $\maxplusN$ that
  computes the clique number on a lower triangular grid graph.  The run of the
  WTS will guess a subset of vertices of the original graph (corresponds to
  labeling some diagonal elements with state $\qcr$) and checks that there is an
  edge between every pair of these (corresponds to checking the label is 1, if
  the row and column end in a $\qcr$-labeled vertex).  The weight of such a run
  will be the size of the subset, and the $\max$ over all the runs gives us the
  clique number as required.

  Let $ Q = \{ \qcr, \qcol, \qrow, \qnone\}$.  A run will
  label a subset of diagonal vertices with $\qcr$.  A vertex is labeled with
  $\qcol$ (resp.\ $\qrow$, $\qcr$) if its column (resp.\ row, both) starts in a
  vertex labeled $\qcr$.  In addition a vertex may get state $\qcr$ only if its
  label is $1$.  All other vertices get state $\qnone$.
  A run on the graph in Figure~\ref{fig:clique-matrix} is depicted in
  Figure~\ref{fig:clique-run}.

  Tiles for diagonal vertices are given by $\Tiles_{(\emptyset,\Edgesout)} =
  \{(\fempty, \qnone, 1, \fout), (\fempty, \qcr, 1, \fout) \} \,.$
  For an inside vertex we have ($\fout$ being arbitrary in all tuples):
  \begin{align*}
    \Tiles_{(\{\eright, \edown\}, \Edgesout)} &=
    \{(\fin, \qnone, b, \fout )\mid b \in \{0,1\}, \fin(\eright)\in\{\qcol,\qnone\},
    \fin(\edown)\in\{\qrow,\qnone\}\}\\
    &\cup\{ (\fin, \qcr, 1, \fout )\mid  \fin(\eright) \in \{\qcr,  \qrow\} , 
    \fin(\edown) \in  \{ \qcr , \qcol \}\}\\
    &\cup\{ (\fin, \qrow, b, \fout )\mid  b \in \{0,1\},\fin(\eright) \in  \{\qcr , \qrow\} , 
    \fin(\edown)\in\{\qrow,\qnone\}\}\\
    &\cup\{ (\fin, \qcol, b, \fout )\mid  b \in \{0,1\}, \fin(\eright)\in\{\qcol,\qnone\}, 
    \fin(\edown) \in   \{\qcr ,\qcol \}\} \,.
  \end{align*}
  
  The weight of a tile of the form $(\fempty, \qcr, 1, \fout)$ is $1$.  Notice
  that only the diagonal vertices labeled $\qcr$ will get such a tile.  The
  weight of all other tiles is $0$.
  Thus the weight of a run is the number of diagonal vertices labeled $\qcr$ -
  which corresponds to a subset of vertices inducing a clique.  The maximum
  weight across different runs will compute the clique number as required.
  \qed
\end{example}

\begin{example}[A WTS to compute the permanent of a (0,1)-matrix]\label{ex:permanent}
  We will model (0,1)-matrices as (0,1)-labelled grids.  As in
  Example~\ref{ex:clique-number}, we let $\Edgenames = \{\eright, \edown\}$ and
  $\Nodelabels = \{0, 1\}$.  A $5 \times 5$ (0,1)-matrix as a grid graph is
  illustrated in Figure~\ref{fig:permanent-matrix}.

  \smallskip
  
  We will define a WTS $\wts$ on such graphs over \Natural\ such that
  $\sem\wts(G)$ is the permanent of the 0,1 matrix $A$ represented by $G$.  In
  each run exactly one vertex in each row and each column will be circled --
  representing one permutation $\sigma$ of $\{1, \dots ,n \}$ if $G$ is an $n
  \times n $ grid.  The weight of the tile on the circled vertex will be the
  vertex label (0 or 1) interpreted as an integer.  Every other tile will have
  weight 1.  Thus the weight of a run will be $\prod_{i} A(i, \sigma(i))$ where
  $\sigma$ is the permutation represented by the run.  Finally the value of a
  graph $G$ representing an $n \times n$ $(0,1)$-matrix $A$ will be
  $\sum_{\sigma} \prod_{i} A(i, \sigma(i))$ which is its permanent.

  The WTS $\wts$ has five states: $Q = \{\qthis, \qur, \qdr, \qdl, \qul\}$.  We
  will define tiles so as to accept only the labeling reflecting the following:
  \begin{itemize}[nosep]
    \item a vertex labeled $\qthis$ means it is the circled vertex in its row and 
    column,
    
    \item a vertex $v$ labeled $\qur$ means that the circled vertex in its column is
    upward of $v$, and the circled vertex in its row is to the right of $v$,
    
    \item similarly for other states $\qdr, \qdl, \qul$.
  \end{itemize}
  The tiles are given formally below.
  The weight function $\Weight$ assigns weight $0$ to any tile labeling a $0$
  labeled node with $\qthis$.  The weight of all other tiles is $1$.
  A run of this WTS is illustrated in Figure~\ref{fig:permanent-run}.
  
\medskip\noindent 
\begin{minipage}[t]{6.5cm}
\begin{center}
\gusepicture[scale = 0.9]{gpic:permanent-matrix}
\end{center}
\captionof{figure}{A $5 \times 5$ (0,1)-matrix as a grid graph}
\label{fig:permanent-matrix}
\end{minipage}
\hfill
\begin{minipage}[t]{7.1cm}
\begin{center}
\gusepicture[scale = 0.9]{gpic:permanent-run}
\end{center}
\captionof{figure}{A run of the WTS $\wts$ on the graph in
Fig.~\ref{fig:permanent-matrix}.  It has weight $0$ as two tiles have
$\Weight$ 0.}
\label{fig:permanent-run}
\end{minipage}

  \medskip
  We now describe the tiles formally.
  For the top-left vertex we have 
  \begin{align*}
    \Tiles_{(\emptyset, \{\eright, \edown\})} &= \{
    (\fempty, \qthis, b,\fout) \mid b \in \{0,1\}, \fout(\eright) = {\qdl}, \fout(\edown) = {\qur}\}\\
    &  \cup \{(\fempty,\qdr, b, \fout) \mid b \in \{0,1\}, \fout(\eright)\in 
    \{\qdr, \qthis\}, \fout(\edown) \in \{\qdr, \qthis\} \}
  \end{align*}
  The tiles for other corner vertices are analogous. 
  For the left border vertices we have $\Tiles_{(\{\edown\}, \{\eright, \edown\})} =$
  \begin{align*}
    & \{(\fin, \qthis, b,\fout) \mid b \in \{0,1\}, \fin(\edown) = {\qdr}, 
    \fout(\eright)\in\{\qdl,\qul\}, \fout(\edown) = {\qur} \} \\
    \cup~ 
    & \{(\fin,\qdr, b, \fout) \mid b \in \{0,1\}, \fin(\edown) = {\qdr},  
    \fout(\eright) \in \{\qur,\qdr,\qthis\}, \fout(\edown) \in \{\qdr,\qthis\} \} \\
    \cup~
    & \{(\fin,\qur, b, \fout) \mid b \in \{0,1\}, \fin(\edown) \in \{\qur,\qthis\},  
    \fout(\eright) \in \{\qur,\qdr,\qthis\}, \fout(\edown) = {\qur} \} 
  \end{align*}
  The tiles for other border vertices are analogous.
  For an interior vertex, we have 
  \begin{align*}
    \Tiles_{(\{\eright, \edown\}, \{\eright, \edown\})} &= 
    \{(\fin, \qthis, b,\fout) \mid b \in \{0,1\}, \fin(\edown) \in \{\qdl,\qdr\},  
    \fin(\eright) \in \{\qur,\qdr\}, 
    \\
    &\hspace{35mm} \fout(\eright) \in \{\qul,\qdl\}, \fout(\edown) \in \{\qul,\qur\} \} 
    \\
    &\cup \{(\fin,\qdr, b, \fout) \mid b \in \{0,1\}, \fin(\edown) \in \{\qdl,\qdr\},
    \fin(\eright) \in \{\qur,\qdr\},
    \\
    &\hspace{35mm} \fout(\eright)\in\{\qur,\qthis,\qdr\}, \fout(\edown)\in\{\qdl,\qthis,\qdr\} \} 
    \\
    &\cup \{(\fin,\qur, b, \fout) \mid b \in \{0,1\}, \fin(\edown)\in\{\qul,\qthis,\qur\},
    \fin(\eright) \in \{\qur,\qdr\},
    \\
    &\hspace{35mm} \fout(\eright) \in \{\qur,\qthis,\qdr\}, \fout(\edown) \in \{\qul,\qur\} \} 
    \\
    &\cup \{(\fin,\qul, b, \fout) \mid b \in \{0,1\}, \fin(\edown)\in\{\qul,\qthis,\qur\},
    \fin(\eright) \in \{\qul,\qthis,\qdl\},
    \\
    &\hspace{35mm} \fout(\eright) \in \{\qul,\qdl\}, \fout(\edown)\in\{\qul,\qur\} \} 
    \\
    &\cup \{(\fin,\qdl, b, \fout) \mid b \in \{0,1\}, \fin(\edown)\in\{\qdl,\qdr\},
    \fin(\eright) \in\{\qul,\qthis,\qdl\},
    \\
    &\hspace{35mm} \fout(\eright) \in \{\qul,\qdl\}, \fout(\edown)\in\{\qdl,\qthis,\qdr\} \} 
  \end{align*}
  Finally, we describe the weight function $\Weight$.
  The weight of a tile of the form $(\fin, \qthis, 0, \fout)$ is $0$. The weight of all other tiles is $1$. 
  \qed
\end{example}

\begin{example}[Permanent of matrix with entries from $\N$]\label{ex:permanent-natural}
  The purpose of this example is to illustrate that it is possible to encode
  natural numbers, which may appear as matrix entries or edge weights, also as
  bounded degree graphs with a fixed alphabet $\Sigma$.

  A length $k$ bit string $ b_{k-1}\cdots b_1b_0$ where $b_i\in\{0,1\}$ for all
  $0\leq i<k$, is represented by a path graph of length $k$.  The vertices of
  this path graph are labelled with $1$ or $0$ to indicate the value of the bit,
  and the edges are labeled $\prec$.  We describe a WTS on such path graphs
  whose computed weight is the binary number $\sum_i b_i2^{i}$.  The WTS guesses
  a prefix ending with label $1$.  All the nodes in the prefix take state $q_0$
  and all nodes after the prefix may take the two states $q_1$ or $q_2$.  The
  weight of all tiles is $1$.  The number of runs is $\sum_{i: b_i=1} 1^{k-i}
  \times 2^i = \sum_i b_i2^i$.

  As before, we will have an $n\times n$ grid graph to represent the matrix, but
  the vertices of the grid graph take a neutral label, say $X$.  A path graph
  originates from every vertex of the grid graph indicating the entry of the
  matrix at that cell.  Now, to compute the permanent, the path graphs starting
  from a circled vertex can start the WTS described in the previous paragraph.
  All other path graphs vertices can be labeled only by a special state $q_4$.
  The weights of all permissible tiles are $1$.  The weight computed by one
  permutation will indeed be the product of the entries.  This crucially depends
  on the distributivity of the semiring.  Thus, this WTS computes the permanent
  of an arbitrary matrix with entries in $\N$.
  \qed
\end{example}

\subparagraph{Evaluation problem (\fval)}
is to compute $\sem\wts(G)$, given the following input:
\begin{center}
\begin{tabular}{l  l}
$\wts$ &: a WTS over $(\Edgenames,\Nodelabels)$-graphs and a semiring $\SR$, and\\
$G$ &: a $(\Edgenames,\Nodelabels)$-graph.
\end{tabular}
\end{center}
We study the complexity of this problem in
Section~\ref{sec:complexity-arbitrary-graphs}, for various semirings.  We
provide an efficient algorithm for this problem in the case of bounded
tree-width graphs in Section~\ref{sec:complexity-btw-graphs}.  In
Section~\ref{sec:discussion} we discuss the decision variants of the above
problem.

\section{Evaluation complexity: Arbitrary graphs}
\label{sec:complexity-arbitrary-graphs}

Recall that we only consider the boolean semiring, the counting semirings over
$\N$, $\Z$, $\Q$ or $\Q_{\geq 0}$ and the tropical semirings over $\N$ or $\Z$.

Given a WTS $\wts$ and a graph $G$, we can compute $\sem{\wts}(G)$ in polynomial
space as follows.  Initialise the current aggregate to $\szero$.  Enumerate in
lexicographic order through the different labelings of the vertices of $G$ with
states of $\wts$.  For each labeling, if it conforms to $\Delta$, compute its
weight and add to the current aggregate.  Thus $\fval$ belongs to \FPSPACE\ --- the set of functions computable in polynomial space.
 
\begin{theorem}
  Problem $\fval$ is in \FPSPACE. 
\end{theorem}

However, for particular semirings the complexity is different as stated in the 
following subsections. 

\subsection{$(+,\times)$-semirings}

\begin{theorem}\label{thm:complexity}
  The evaluation problem is $\sharpP$-complete over \Natural, and non-negative
  \Rational.  It is $\gapP$-complete over \Integer\ and \Rational.
\end{theorem}

The upper bounds hold for arbitrary graphs, and the lower bounds hold for the
special case of grids.  The weights can be assumed to be given in binary.

\smallskip

A function $f$ is in $\#$\Ptime\ if there is an \NP\ machine $M$ such that $f(x)
= \#M(x)$.  That is, it denotes the set of function problems that correspond to
counting the number of accepting paths in a non-deterministic polynomial time
turing machine.  Computing the permanent of a (0,1)- matrix is a
$\sharpP$-complete problem \cite{Valiant79}, and hence the $\sharpP$-hardness
claimed above follows from Example~\ref{ex:permanent}.  We give an alternate
hardness proof by a reduction from $\#$-CNF-SAT.
 
A function $f(x)$ is in \gapP\ if there is a non-deterministic polynomial time
turing machine $M$ such that $f(x) = \#M(x)-\#\overline M(x)$.
\gapP\ is also the closure of $\#$P under subtraction.

\smallskip

Most of this subsection is devoted to the proof of Theorem~\ref{thm:complexity}.
First we give the non-deterministic Turing machines realising the upper bounds
for \Natural\ and \Integer.  After that we give reductions from respective
counting versions of SAT to prove the lower bounds.  
The case of \Rational\ is finally considered.

\subparagraph{The Turing Machine $\ma$ such that $\#\ma(\wts,G) = \sem{\wts}(G)$:}
We describe a non-determinsitc polynomial time turing machine $\ma$ that takes as input 
a WTS $\wts$ over \Natural\ with weights given in binary, and a graph $G$.  
The number of accepting runs $\#\ma(\wts,G) = \sem{\wts}(G)$.  We assume the
states, weights etc.  are given by some standard encoding.
 
The turing machine $\ma$ non-deterministically guesses a labeling of the
vertices of $G$ by the states of $\wts$.  Then it computes the product $w$ of
the weights of the tiles in the guessed tiling and writes it in binary (MSB on
the left) in a different tape.  Computing the product can be done in time
polynomial in $|G|$ and $\log(k)$ where $k = \max \{x \mid x \text{ is a weight of
some tile of } \wts\}$.

Afterwards it enters a phase which will have exactly $w$ different accepting
branches.  Simply decrementing the value while it is positive, and
non-deterministically accepting at any step will have $w$ accepting branches,
but the running time is exponential.  We want the machine to run in polynomial
time.  Hence we implement this phase similar to
Example~\ref{ex:permanent-natural}.  It runs in $\mathcal{O}(|w|)$ steps as we
detail below.

$\ma$ scans $w$ from left to right starting in some state $q$.  While in
state $q$ and the current cell is labeled $0$ it moves right.  If in state $q$
and the current cell is labelled $1$ it moves right and non determistically
stays in state $q$ or enters one of the two special states $q_0$ or $q_1$.  When
it is in state $q_0$ or $q_1$ and the current cell is labelled with $0$ or $1$,
it will move right and non deterministically chose either $q_0$ or $q_1$.
Finally, When in state $q_0$ or $q_1$ and the current cell is \textit{blank}
(i.e., the scan of $w$ is over), then $\ma$ accepts.  Thus if the $i$th bit
from the right of $w$ is labeled $1$, then $\ma$ can have $2^{i}$ accepting
runs if it moved from state $q$ to $q_0$ or $q_1$ when reading this bit.
Switching from state $q$ can occur at any $1$-labelled cell, and hence
$\ma$ will have $w$ many accepting runs.

The machine $\ma$ non deterministically picks a labeling at first, and hence the
total number of accepting runs $\#\ma(\wts,G) = \sem{\wts}(G)$.  With this we
prove the $\sharpP$ upper bound for \Natural.

\subparagraph{The Turing Machine $\ma'$ such that
$\#\ma'(\wts,G) - \#\overline{\ma'}(\wts,G) = \sem{\wts}(G)$:}
This is similar to the machine $\ma$ above.  There are two differences.  The
machine $\ma'$ still guesses a labeling of vertices of $G$ with states of $\wts$ over \Integer\ 
and computes the weight $w$.  If $w$ is positive, it proceeds exactly as $\ma$
does to produce $w$ accepting runs.  If the weight $w$ is negative, the machine
$\ma'$ proceeds analogously but with states $q'$, $q_0'$ and $q_1'$ instead.  If
the machine is in state $q'_0$ or $q'_1$ with current cell \textit{blank} then
it rejects instead of accepting.  The second difference is for blocked runs
(e.g., if the guessed labeling of vertices of $G$ by states of $\wts$ is not a
valid tiling, or if at the end the machine is still in state $q$ or $q'$ with
current cell \textit{blank}).  In such a case, $\ma'$ will non-deterministically
proceed to either accept or reject.  Thus the net difference between accepting
runs and rejecting runs is kept intact and $\#\ma'(\wts,G) -
\#\overline{\ma'}(\wts,G) = \sem{\wts}(G)$. This proves the $\gapP$ upper bound for \Integer.

\newcommand{\gphi}{G_\varphi}
\subparagraph{Encoding a CNF formula $\varphi$ in a grid $\gphi$:}
Given a CNF formula $\varphi$ with $n$ variables and $m$ clauses, we encode it
in an $n \times m $ grid with node labels $\{p, n, \star\}$.  If the node $(i, j)$
is labeled by $p$ (resp.\  $n$) it means that the $i$th variable appears in $j$th
clause positively (resp.\  negatively).  The node $(i,j)$ is labeled $\star$ if
the $i$th variable does not occur in the $j$th clause.

\newcommand{\wtscount}{\wts^\#}
\newcommand{\qtrue}{q_\textsf{true}}
\newcommand{\qfalse}{q_\textsf{false}}
\newcommand{\qptrue}{q'_\textsf{true}}
\newcommand{\qpfalse}{q'_\textsf{false}}

\subparagraph{A WTS $\wtscount$ over \Natural\ for counting $\#\varphi$:}
Recall that $\#\varphi$ is the number of satisfying assignments for the formula 
$\varphi$.
We assume input to the WTS $\wtscount$ is given as $\gphi$ --- a
$\{p,n,\star\}$-labeled grid encoding a CNF formula.

A state of $\wtscount$ is a pair from
$\{\qtrue,\qfalse\}\times\{\qptrue,\qpfalse\}$.  The first part of a state
indicates a truth assignment with $\qtrue$ and $\qfalse$.  The allowed tiles
make sure that in this part the truth assignment remains the same along a row.
The second part of a state
indicates with $\qptrue$ and $\qpfalse$ the partial evaluation of the formula.
A $p$-labeled node which is assigned $\qtrue$ from the first part, and an
$n$-labeled node which is assigned $\qfalse$ from the first part gets the value
$\qptrue$ in the second part of the state (call this condition A for future
reference).  Further all the successor nodes in the column of the $\qptrue$
labeled node also gets the value $\qptrue$, except for the nodes in the last
row.  For the nodes in the last row, it gets the value $\qptrue$ if the left
neighbour is labeled $\qptrue$ (assume this is satisfied if the left neighbour
does not exist), and a) if it satisfies condition A or b) if the node above is
labeled $\qptrue$.  Otherwise the nodes get the value $\qpfalse$.  The second
part of a state labeling a node $(n, j)$ in the last row indicates the
evaluation of the prefix of the formula until the $j$th clause.

The tiles capture the description above.
The weight of all tiles is $1$, except for the tile labeling
the last node $(n,m)$.  If it is labeled $(-, \qptrue)$ then the weight is $1$,
otherwise it is $0$.  The value $\sem{\wtscount}(\gphi) = \#\varphi$, the number
of satisfying assignments.

This proves the $\sharpP$ lower bound for \Natural.  As alluded to earlier, the
permanent computation (Example~\ref{ex:permanent}) gives an alternate lower
bound proof.

\newcommand{\wtsgap}{\wts^\textsf{gap}}
\newcommand{\qskip}{q_\text{skip}}

\subparagraph{A WTS $\wtsgap$ over \Integer\ for counting $\#\varphi_1 - \#\varphi_2$:}
We will reduce the $\gapP$-complete problem of computing $\#\varphi_1 -
\#\varphi_2$, where $\varphi_1$ and $\varphi_2$ are input CNF formulas on the
same set of $n$ variables with $m_1$ and $m_2$ clauses respectively.  We
represent the input in an $n \times (m_1 + m_2)$ grid by putting $G_{\varphi_1}$
and $G_{\varphi_2}$ side by side.  The node labels contain a special tag $i \in
\{1,2\}$ to indicate that it comes from $G_{\varphi_i}$.  The WTS $\wtsgap$ will
ensure that rows are of the form $1^\ast2^\ast$ and columns are of the form
$1^\ast$ or $2^\ast$.  In a run it evaluates either $\varphi_1$ or $\varphi_2$
similar to $\wtscount$.  If it is evaluating $\varphi_i$ all nodes with the tag
$3-i$ gets a special state $\qskip$.  The weight of all tiles is $1$, except for
the tile labeling the nodes $(n,m_1)$ and $(n,m_1 + m_2)$.
If the node $(n,m_1)$ is labeled $(-, \qptrue)$ or $\qskip$ then the weight is
$1$, otherwise it is $0$.  If the node $(n,m_1 + m_2)$ is labeled $(-, \qptrue)$
(resp.\ $\qskip$) then the weight is $-1$ (resp.\ $1$), otherwise it is $0$.

\subparagraph{\Rational}
We will use counting reduction from $\Rational$ (resp.\ {non-negative}
$\Rational$) to the evaluation problem over \Integer\ (resp.\ $\Natural$) in
order to prove the upper bounds.  First we will transform an input $(\wts, G)$
of the evaluation problem over \Rational\ (resp.\ {non-negative} $\Rational$) to
an input $(\wts', G, )$ over $\Integer$ (resp.\ \Natural).  In $\wts'$ we will
multiply the weight of a tile by $\ell$ - the lcm of the denominators appearing
in the weights of any tile of $\wts$.  The multiplication can be performed in
time polynomial.
Now $\wts'$ is a WTS over \Integer\ (resp.\ \Natural), and following the $\gapP$
procedure (resp.\ $\sharpP$ procedure) we compute $\sem{\wts'}(G)$.  Now, we
transform the output back to the required output over $\Rational$ (resp.\
{non-negative} $\Rational$) by dividing with $\ell^{|V_G|}$.  That is, $\fval(G,
\A) = \frac{\fval(G, \A')}{\ell^{|V_G|}}$.

Notice that we allow the weights to be given in binary.  The lcm $\ell$ and
$\ell^{|V_G|}$ can be computed in polynomial time.  The counting reduction is
hence polynomial.  This proves the upper bounds.

\smallskip

The $\gapP$-hardness (resp.\ $\#\Ptime$-hardness) follows because \Integer\
(resp.\ $\Natural$) is a special case of $\Rational$ (resp.\ {non-negative}
$\Rational$).

\newcommand{\memb}{\textsf{Membership}}
\subsection{\Boolean\ semiring.}
Note that the evaluation problem $\fval$ over $\Boolean$ is in fact the
classical Membership problem (denoted $\memb$) and is indeed a decision problem.
We can check in \NP\ whether the value is $1$ (witnessed by the \NP\ machine
$\ma$, if the input is assumed to be over $\Boolean$ then $\times$ serve as
$\wedge$).  It is also \NP-hard by a simple reduction from CNF SAT (witnessed by
$\wts^\#$ interpreted over \Boolean).

\begin{theorem}\label{thm:boolean}
	$\memb$ is \NP-complete.
\end{theorem}

\subsection{Tropical semirings.}

\begin{theorem}
  We assume the weights are given in unary. 
  The evaluation problem over any tropical semiring is $\FPtoNPlog$-complete.
\end{theorem}

$\FPtoNPlog$ is the class of functions computable by a polynomial time turing
machine with logarithmically many queries to $\NP$.

\begin{proof}
  We will prove the upper bound for $\maxplusZ$.  The case of $\maxplusN$ is
  subsumed.  The cases of $\minplusN$ and $\minplusZ$ are analogous.

  Let $k$ be the maximal constant and $\ell$ be the minimal constant (other than
  $+/-\infty$) appearing in the WTS $\A$.  The maximum possible weight of a run
  is $n \times k$ and the minimum is $n \times \ell$ where $n$ is the number of
  vertices in the input graph.  We will do a binary search in the set $ W = \{n
  \times \ell, \dots, -1, 0, 1, \dots, n \times k \}$ checking if $\sem{\A}(G)
  \ge s$ to find the value of $\sem{\A}(G) $.  In each iteration of the binary
  search, we make an oracle call to the \NP\ machine for $\sem{\A}(G) \ge s$.
  The number of \NP\ oracle queries is $\mathcal{O}(\log (n \times k) )$ which
  is only logarithmic in the input size.  Recall that the weights are encoded in
  unary.

  Finding the clique number is an $\FPtoNPlog$-complete
  problem~\cite{Krentel88}.  From Example~\ref{ex:clique-number}, the lower
  bound follows.
\end{proof}

\section{Efficient evaluation for bounded tree-width graphs}
\label{sec:complexity-btw-graphs}

\newcommand{\add}[2]{\mathop{\textsf{Add}_{#1}^{#2}}}
\newcommand{\forget}[1]{\mathop{\textsf{Forget}_#1}}
\newcommand{\rename}[2]{\mathop{\textsf{Rename}_{#1,#2}}}
\newcommand{\ttunion}{\oplus}
\newcommand{\dom}{\textsf{dom}}
\newcommand{\Bc}{\mathcal{B}}
\newcommand{\val}{\mathsf{val}}

\newcommand{\TT}{\textsf{TT}\xspace}
\newcommand{\TTs}{\textsf{TTs}\xspace}
\newcommand{\kTT}{$k$-\textsf{TT}\xspace}
\newcommand{\kTTs}{$k$-\textsf{TTs}\xspace}
\newcommand{\kTW}{$k$-\textsf{TW}\xspace}

\newcommand{\kPT}{$k$-\textsf{word}\xspace}
\newcommand{\kPTs}{$k$-\textsf{words}\xspace}
\newcommand{\kPW}{$k$-\textsf{PW}\xspace}
\newcommand{\kWords}{\mathsf{W}_k}

In this section, we show that the problem $\fval$ can be solved efficiently when
restricted to graphs of bounded tree-width (the bound is not part of the input).
By efficient, we mean time polynomial wrt.\ the WTS $\wts$ and linear wrt.\ the
graph $G$ (see Theorems~\ref{thm:kPW-feval} and~\ref{thm:kTW-feval} below).
Bounded tree-width covers many graphs used to model behaviours of concurrent or
infinite-state systems.  For example, it is well-known that words and trees have
tree-width 1, nested words used for pushdown systems have tree-width 2,
Mazurkiewicz traces describing behaviours of concurrent asynchronous systems
with rendez-vous, and most decidable under-approximations of Turing complete
models such as multi-pushdown automata, message
passing automata with unbounded FIFO channels, 
etc.~\cite{MadhusudanParlato-POPL11,Aiswarya-PhD14,BolligGastin-mpri2019}.
We start by explaining our results for bounded path-width since this is 
technically simpler. Then we explain how this is extended to bounded tree-width.

\subsection{Bounded path-width evaluation}

\subparagraph{A path decomposition} of a $(\Gamma,\Sigma)$-graph
$G=(V,(E_\gamma)_{\gamma\in\Gamma},\lambda)$, is a sequence $V_1,\ldots,V_n$ of
nonempty subsets of vertices satisfying:
\begin{enumerate}[nosep]
  \item  for all $v\in V$, we have $v\in V_i$ for some $1\leq i\leq n$,

  \item  for all $(u,v)\in\bigcup_{\gamma\in\Gamma}E_\gamma$, we have 
  $u,v\in V_i$ for some $1\leq i\leq n$,

  \item  for all $1\leq i\leq j\leq k\leq n$, we have $V_i\cap V_k\subseteq V_j$.
\end{enumerate}
The width of the path decomposition is $\max\{|V_i|-1\mid 1\leq i\leq n\}$.
The path-width of a graph $G$ is the least $k$ such that $G$ admits a path 
decomposition of width $k$.

Words have path-width 1, but trees, nested words, grids have unbounded
path-width.

\smallskip
We present below an equivalent definition of path-width which will be convenient
to solve the evaluation problem on graphs with bounded path-width.
Let $[k]=\{0,1,\ldots,k\}$.
Graphs over $(\Gamma,\Sigma)$ of path-width at most $k$ can be described with
words over the alphabet
$$
  \Omega_k = \{(i,a) \mid i\in[k],~ a\in\Sigma\} 
  \cup \{\forget{i} \mid i\in[k]\} 
  \cup \{\add{i,j}{\gamma} \mid i,j\in[k],~ \gamma\in\Gamma\}
$$
The semantics of a word $\tau\in\Omega_k^{*}$ is a colored graph
$\wsem\tau=(G_\tau,\chi_\tau)$ where $G_\tau$ is a $(\Gamma,\Sigma)$-labeled
graph and $\chi_\tau\colon [k]\to V$ is a partial injective function coloring
some vertices of $G_\tau$.  We say that a color $i\in[k]$ is \emph{active} in
$\tau$ if it is in the domain of $\chi_\tau$.  The semantics is defined by
induction on the length of $\tau$.  The semantics of the empty word
$\tau=\varepsilon$ is the empty graph.  Assuming that
$\wsem{\tau}=(V,(E_\gamma)_{\gamma\in\Gamma},\lambda,\chi)$, we define the effect
of appending a new letter to $\tau$:
$(i,a)$ adds a new $a$-labeled vertex with color $i$, provided $i$ is not active
in $\tau$, $\forget{i}$ removes color $i$ from the domain of the color map, and
$\add{i,j}{\alpha}$ adds an $\alpha$-labeled edge between the vertices colored $i$
and $j$ (if such vertices exist, i.e., if $i,j$ are active in $\tau$). 
Formally,
\begin{itemize}[nosep]
  \item $\wsem{\tau\cdot(i,a)}=(V',(E_\gamma)_{\gamma\in\Gamma},\lambda',\chi')$
  is defined if $i\notin\dom(\chi)$ and in this case $V'=V\uplus\{v\}$,
  $\lambda'(v)=a$ and $\lambda'(u)=\lambda(u)$ for all $u\in V$,
  $\dom(\chi')=\dom(\chi)\uplus\{i\}$, $\chi'(i)=v$ and $\chi'(j)=\chi(j)$ for
  all $j\in\dom(\chi)$.
  
  \item
  $\wsem{\tau\cdot\forget{i}}=(V,(E_\gamma)_{\gamma\in\Gamma},\lambda,\chi')$
  with $\dom(\chi')=\dom(\chi)\setminus\{i\}$ and $\chi'(j)=\chi(j)$ for all
  $j\in\dom(\chi')$.

  \item
  $\wsem{\tau\cdot\add{i,j}{\alpha}}=(V,(E'_\gamma)_{\gamma\in\Gamma},\lambda,\chi)$
  with $E'_\gamma=E_\gamma$ if $\gamma\neq\alpha$ and
  $$
  E'_\alpha= 
  \begin{cases}
    E_\alpha & \text{if } \{i,j\}\not\subseteq\dom(\chi) \\
    E_\alpha\cup\{(\chi(i),\chi(j))\} & \text{otherwise.}
  \end{cases}
  $$
  \qedhere
\end{itemize}

We say that a word $\tau$ over $\Omega_k$ is \emph{well-formed} if the following conditions 
are satisfied:
\begin{enumerate}[nosep]
  \item  if $\tau'\cdot(i,a)$ is a prefix of $\tau$ then $i$ is not active in $\tau'$,

  \item  if $\tau'\cdot\forget{i}$ is a prefix of $\tau$ then $i$ is active in $\tau'$,

  \item if $\tau'\cdot\add{i,j}{\gamma}$ is a prefix of $\tau$ then $i,j$ are
  active in $\tau'$ and the edge labeled $\gamma$ was not already added
  in $\tau'$ between $\chi_{\tau'}(i)$ and $\chi_{\tau'}(j)$.
\end{enumerate}
In the following, a well-formed word over $\Omega_k$ is called a \kPT.
The set $\kWords\subseteq\Omega_k^{*}$ of \kPTs is clearly regular.

\begin{lemma}\label{lem:path-term-and-decomposition}
  \begin{enumerate}
    \item Given a path decomposition $V_1,\ldots,V_N$ of width at most $k$ of a
    $(\Gamma,\Sigma)$-graph $G$, we can construct in linear time wrt.\ $|G|$ a
    \kPT $\tau$ such that $\wsem{\tau}=(G,\emptyset)$.

    \item Given a $\kPT$ $\tau$, we can construct a path decomposition of width
    at most $k$ of the graph $G_\tau$ defined by $\tau$: 
    $\wsem{\tau}=(G_\tau,\chi_\tau)$.
  \end{enumerate}
\end{lemma}

\begin{proof}
  \textbf{1.}
  We construct by induction a sequence of \kPTs $\tau_\ell$ for $0\leq\ell\leq
  N$ such that $\wsem{\tau_\ell}=(G_\ell,\chi_\ell)$ where $G_\ell$ is the
  subgraph of $G=(V,(E_\gamma)_{\gamma\in\Gamma},\lambda)$ induced by the
  vertices $V_1\cup\cdots\cup V_\ell$, and $\chi_\ell([k])=V_\ell\cap
  V_{\ell+1}$ (with $V_0=V_{N+1}=\emptyset$).  We let $\tau_0=\epsilon$.
  
  Let now $0\leq\ell<N$ and assume that $\tau_{\ell}$ has been constructed.  Let
  $C_\ell=\dom(\chi_\ell)\subseteq[k]$ be the active colors in $\tau_\ell$.  By
  induction, we know that $|C_\ell|=|V_{\ell+1}\cap V_\ell|$.  Let
  $V_{\ell+1}\setminus V_\ell=\{u_1,\ldots,u_m\}$.  Since the decomposition is
  of width at most $k$, we have $|V_{\ell+1}|\leq1+k$ and we find
  $i_1<\cdots<i_m$ available colors in $[k]\setminus C_\ell$.  We define
  $\tau'_{\ell+1}=\tau_\ell\cdot(i_1,\lambda(u_1))\cdots(i_m,\lambda(u_m))$.
  Let $D=\{i_1,\ldots,i_m\}$ and let $\wsem{\tau'_{\ell+1}}=(G',\chi')$.  We have
  $\dom(\chi')=C_\ell\cup D$, $\chi'(C_\ell)=V_{\ell+1}\cap V_{\ell}$ and
  $\chi'(D)=V_{\ell+1}\setminus V_\ell$.
  For each $\gamma\in\Gamma$, $i\in C_\ell\cup D$ and $j\in D$ such that 
  $(\chi'(i),\chi'(j))\in E_\gamma$ (resp.\ $(\chi'(j),\chi'(i))\in E_\gamma$),
  we append $\add{i,j}{\gamma}$ (resp.\ $\add{j,i}{\gamma}$) to the word
  $\tau'_{\ell+1}$.  We obtain a \kPT $\tau''_{\ell+1}$ which defines the
  subgraph $G_{\ell+1}$ of $G$ induced by $V_1\cup\cdots\cup V_{\ell+1}$.
  Notice that, from the third condition of a path decomposition, we have
  $V_{\ell+1}\setminus V_\ell= V_{\ell+1}\setminus(V_1\cup\cdots\cup V_\ell)$
  and the edges in $G_{\ell+1}$ which were not already in $G_\ell$ are between
  some vertex in $V_{\ell+1}\setminus V_\ell$ and some vertex in $V_{\ell+1}$.
  Finally, for each $i\in C_\ell\cup D$ such that $\chi'(i)\notin V_{\ell+2}$,
  we append $\forget{i}$ to the word $\tau''_{\ell+1}$.  We obtain the \kPT
  $\tau_{\ell+1}$ satisfying our invariant.
  
  Finally, from the invariant we deduce that $\wsem{\tau_N}=(G,\emptyset)$,
  which concludes the first part of the proof.
  
  \medskip\noindent
  \textbf{2.}
  Let $\tau$ be a \kPT and $n=|\tau|$ be its length.  For $0\leq\ell\leq n$, let
  $\tau_\ell$ be the prefix of $\tau$ of length $\ell$.  Let
  $\wsem{\tau_\ell}=(G_\ell,\chi_\ell)$ and $V_\ell=\chi_\ell([k])$ be the subset
  of vertices which are colored in $\wsem{\tau_\ell}$.  We show that
  $V_1,\ldots,V_n$ is a path decomposition of
  $G=G_n=(V,(E_\gamma)_{\gamma\in\Gamma},\lambda)$.
  
  Let $u\in V$ be a vertex of $G$.  
  For some $1\leq\ell\leq n$, we have $\tau_{\ell}=\tau_{\ell-1}\cdot(i,a)$ with
  $\chi_\ell(i)=u\in V_\ell$.  This proves that the first condition of a path
  decomposition is satisfied.
  
  Let $(u,v)\in E_\gamma$ for some $\gamma\in\Gamma$. For some $1<\ell<n$, we 
  have $\tau_{\ell+1}=\tau_\ell\cdot\add{i,j}{\gamma}$ with $\chi_\ell(i)=u$ 
  and $\chi_\ell(j)=v$. We deduce that $u,v\in V_\ell$, which proves that the 
  second condition of a path decomposition is satisfied.
  
  For the third condition, let $1\leq i\leq j\leq m\leq n$ and $u\in V_i\cap 
  V_m$. We deduce that for some $\ell\in[k]$, we have 
  $u=\chi_i(\ell)=\chi_m(\ell)$ and that color $\ell$ was not forgotten between 
  $\tau_i$ and $\tau_m$. Therefore, $u=\chi_j(\ell)\in V_j$ as desired.  
\end{proof}

\subparagraph{Existentially bounded graphs.}
Another characterization of bounded path-width is the notion of
\emph{existentially-bounded} graphs~\cite{Aiswarya-PhD14,BolligGastin-mpri2019}. 
Let $G=(V,(E_\gamma)_{\gamma\in\Gamma},\lambda)$ be a $(\Gamma,\Sigma)$ graph
and $k>0$ an integer.  The graph $G$ is \emph{existentially $k$-bounded}
($\exists k$-bounded) if there is a linear order $<$ on the vertices of $G$ such
that for all $v\in V$, the number of vertices $u\leq v$ connected to some
vertices $w>v$ is at most $k$:
\begin{align}
  &|\{u\in V\mid u\leq v \text{ and } 
  (u,w)\in\textstyle\bigcup_{\gamma\in\Gamma}E_\gamma\cup E_\gamma^{-1}
  \text{ for some } w>v\}|\leq k \,.
  \label{eq:k-bounded}
\end{align}
A linear order $<$ satisfying \eqref{eq:k-bounded} is called 
$k$-bounded\footnote{Notice that the bound is on the number of vertices and not 
on the number of edges $(u,w)$ crossing over $v$, i.e., with $u\leq v<w$. But 
when the graph has bounded degree, the  bound on the number of vertices induces 
a bound on the number of crossing edges.}.

Words are $\exists 1$-bounded.  Mazurkiewicz traces, with the process based
representation, are $\exists K$-bounded where $K$ is the number of processes.
Message sequence charts (MSCs) are not existentially bounded in general.  But
existentially bounded MSCs is a fundamental well-behaved under-approximation for
message passing automata with unbounded FIFO channels.

\begin{example}\label{ex:grid-exists-bounded}
  An $m\times n$ grid is existentially $\min(m,n)$-bounded.  Indeed, the set of
  vertices of an $m\times n$ grid is $V=\{1,\ldots,m\}\times\{1,\ldots,n\}$.  We
  have horizontal and vertical edges:
  \begin{align*}
    E_\to &= \{((i,j),(i,j+1))\mid 1\leq i\leq m, 1\leq j<n\} \\
    E_\downarrow &= \{((i,j),(i+1,j))\mid 1\leq i<m, 1\leq j\leq n\} 
  \end{align*}
  Assuming that $n\leq m$, we define a linear order on $V$ by listing the first
  row, then the second row, etc.
  $$
  (1,1)<\cdots<(1,n)<(2,1)<\cdots<(2,n)<\cdots<(m,1)<\cdots<(m,n) \,.
  $$
  It is easy to check that this linear order is $n$-bounded.
  If $m<n$ then we list the vertices column by column.
\end{example}

\begin{lemma}\label{lem:exists-k-bounded}
  A graph $G$ is $\exists k$-bounded if and only if its path-width is at most $k$.  
\end{lemma}

\begin{proof}
  Let $G=(V,(E_\gamma)_{\gamma\in\Gamma},\lambda)$ be a $(\Gamma,\Sigma)$ graph
  and let $\tau$ be a \kPT with $\wsem{\tau}=(G,\chi)$.  Let $<$ be the linear
  order on $V$ corresponding to the order in which the vertices of $G$ are
  created by $\tau$.  We show that $<$ is $k$-bounded.  Let $v\in V$ be a
  vertex.  We have to show that there are at most $k$ vertices $u\leq v$ which
  are connected to some vertex $w>v$.  This is obvious if $v$ is one of the
  first $k$ vertices wrt.\ the linear order $<$.  Let $\tau'$ be the prefix of
  $\tau$ which ends in the creation of $v$: $\tau'=\tau''\cdot(\ell,a)$,
  $\wsem{\tau'}=(G',\chi')$ and $\chi'(\ell)=v$.  Let $u\leq v<w$ with $(u,w)\in
  E_\gamma\cup E^{-1}_\gamma$ for some $\gamma\in\Gamma$.  Let $i,j$ be the
  colors of $u$ and $w$ when they were respectively added.  The edge between $u$
  and $w$ was added with $\add{i,j}{\gamma}$ or $\add{j,i}{\gamma}$ after the
  creation of $w$, hence after the prefix $\tau'$ of $\tau$.  We deduce that
  color $i$ was not forgotten at $\tau'$ and $i\in\dom(\chi')$.  Therefore, the
  set $U$ of vertices $u\leq v$ which are connected to some vertex $w>v$ is
  contained in $\chi'([k])$.  If $|\dom(\chi')|\leq k$ then $|U|\leq k$ and we
  are done.  Otherwise, we have $\dom(\chi')=[k]$.  Assume $U\neq\emptyset$ and
  let $u\leq v<w$ with $(u,w)\in E_\gamma\cup E^{-1}_\gamma$ for some
  $\gamma\in\Gamma$.  Let $j$ be the color of $w$ when it was created.  Since
  $j\in\dom(\chi')$, we must see $\forget{j}$ between the creation of $v$ at
  $\tau'$ and the creation of $w$.  Let $\tau''$ be the least prefix of $\tau$
  which ends with some $\forget{m}$ operation and such that $\tau'$ is a prefix
  of $\tau''$.  With $\wsem{\tau''}=(G'',\chi'')$, we have
  $\dom(\chi'')=[k]\setminus\{m\}$ and $\chi''([k])$ contains exactly $k$
  vertices.  As above, we can show that $U\subseteq\chi''([k])$, which concludes
  this direction of the proof.
  
  \medskip
  Let $G=(V,(E_\gamma)_{\gamma\in\Gamma},\lambda)$ be a $(\Gamma,\Sigma)$ graph
  and let $<$ be a linear order on $V$ which is $k$-bounded.  We assume that
  $V=\{v_1,\ldots,v_n\}$ with $v_1<\cdots< v_n$.  For each $0\leq \ell\leq n$, we
  construct by induction a \kPT $\tau_\ell$ which describes the subgraph of
  $G$ induced by the vertices $\{v_1,\ldots,v_\ell\}$, keeping colors on vertices
  $v_i$ which are still missing some edges, i.e.,
  $(v_i,v_j)\in\bigcup_{\gamma\in\Gamma}E_\gamma\cup E_\gamma^{-1}$ with $1\leq
  i\leq \ell<j\leq n$. We start with $\tau_0=\varepsilon$.
  
  Now, let $1\leq\ell\leq n$ and assume that $\tau_{\ell-1}$ is already defined.
  Let $C_{\ell-1}\subseteq[k]=\{0,1,\ldots,k\}$ be the set of active colors in
  $\tau_{\ell-1}$.  Since the linear order $<$ is $k$-bounded, we have
  $|C_{\ell-1}|\leq k$ and we let $i=\min([k]\setminus C_{\ell-1})$.  We add the
  new vertex $v_\ell$ by considering
  $\tau'_\ell=\tau_{\ell-1}\cdot(i,\lambda(v_\ell))$.  Then, we obtain
  $\tau''_\ell$ by adding all edges connecting $v_\ell$ with earlier vertices.
  Assume that $(v_m,v_\ell)\in E_\gamma$ (resp.\ $(v_\ell,v_m)\in E_\gamma$) for
  some $1\leq m<\ell$ and $\gamma\in\Gamma$.  Then, from our invariant, $v_m$
  has some color $j\in C_{\ell-1}$ in $\tau_{\ell-1}$ (hence also in
  $\tau'_\ell$).  We append to the current word $\add{j,i}{\gamma}$ (resp.\
  $\add{i,j}{\gamma}$).  Finally, we obtain $\tau_\ell$ from $\tau''_\ell$ by
  forgetting colors of completed vertices.  If $v_m$ is the vertex corresponding
  to some active color $j\in C_{\ell-1}\cup\{i\}$ and $v_m$ is connected only to
  vertices in $\{v_1,\ldots,v_\ell\}$, then we append $\forget{j}$ to the
  current word.  Notice that the set $C_\ell$ of active colors in $\tau_\ell$
  satisfies $C_\ell\subseteq C_{\ell-1}\cup\{i\}$ and $|C_\ell|\leq k$ since $<$
  is $k$-bounded.
  
  Finally, we have $\wsem{\tau_n}=(G,\emptyset)$, which completes the proof.
  Notice that the construction of the $\kPT$ $\tau$ from the $k$-bounded
  linearization $<$ takes linear time.
\end{proof}

\subparagraph{Solving the evaluation problem in polynomial time.}
The problem \kPW-\feval\ is to compute $\sem{\wts}(G)$, given a WTS $\wts$ and a
$(\Gamma,\Sigma)$-graph $G$ of path-width at most $k$.

\begin{theorem}\label{thm:kPW-feval}
  The problem \kPW-\feval\ can be solved in linear time wrt.\ the input graph
  $G$ and polynomial time wrt.\ the input WTS $\wts$.
\end{theorem}

\begin{proof}
  The evaluation algorithm for bounded path-width graphs proceeds in three steps:
  \begin{enumerate}[nosep]
    \item  From the input graph $G$, which is assumed to be of path-width at most 
    $k$, we compute in linear time a path decomposition $V_1,\ldots,V_n$ using 
    Bodlaender's algorithm \cite{Bodlaender_1996}. Then, using 
    Lemma~\ref{lem:path-term-and-decomposition}, we compute in linear time a \kPT 
    $\tau$ such that $\wsem{\tau}=(G,\emptyset)$.
  
    \item By Lemma~\ref{lem:WTS-2-WWA-k} below, we construct in time polynomial
    in $\wts$ a weighted word automaton $\Bc_k$ which is equivalent to $\wts$ on
    graphs of path-width at most $k$.  
    Equivalent means that for all \kPTs $\tau$ with
    $\wsem{\tau}=(G,\emptyset)$, we have $\sem{\wts}(G)=\sem{\Bc_k}(\tau)$.
  
    \item We compute $\sem{\Bc_k}(\tau)$.  It is well-known that the value of a
    weighted word automaton $\Bc$ on a given word $w$ can be computed in time
    $\mathcal{O}(|\Bc|\cdot|w|)$ assuming that sum and product in the semiring
    take constant time. 
    For the sake of completeness, we give details in 
    Lemma~\ref{lem:weighted-word-automaton-evaluation}. Alternatively, we may 
    use Algorithm~\ref{alg:tree-eval} which achieves the same complexity in the 
    more general case of weighted tree automata.
    \qedhere
  \end{enumerate}
\end{proof}

A weighted word automaton over alphabet $\Sigma$ is usually given as a tuple
$\Bc=(Q,T,I,F,\Weight)$ where $I,F\subseteq Q$ are the subsets of initial
and final states, $T\subseteq Q\times\Sigma\times Q$ defines the transitions
and $\Weight\colon T\to S$ gives weights to transitions.  This is an
equivalent representation of a WTS over $(\{\to\},\Sigma)$.

\begin{lemma}\label{lem:WTS-2-WWA-k}
  Given a WTS $\wts$ over $(\Edgenames,\Nodelabels)$-graphs and $k>0$, we can
  compute in polynomial time wrt.\ $\wts$, a weighted word automaton $\Bc_k$ which
  is equivalent to $\wts$ over graphs of path width at most $k$.
  That is, for all \kPTs $\tau$ with $\wsem{\tau}=(G,\emptyset)$, we
  have $\sem{\wts}(G)=\sem{\Bc_k}(\tau)$.
\end{lemma}

\begin{proof}
  Let $\wts = (\States, \Tiles, \Weight)$ be a WTS over
  $(\Edgenames,\Nodelabels)$-graphs.  By adding tiles with weight $\szero$, we 
  may assume wlog that $\Tiles$ contains all possible tiles. Fix $k\geq1$. 

  A state of $\Bc_k$ is a partial map $\delta\colon[k]\to\Tiles$.  When reading
  a \kPT $\tau$ with $\wsem{\tau}=(G,\chi)$, the automaton will guess a labelling
  $\rho\colon V\to\States$ of vertices of $G$ with states of $\wts$ and will reach
  a state $\delta$ satisfying the following two conditions:
  \begin{enumerate}[nosep]
    \item $\dom(\delta)=\dom(\chi)\subseteq[k]$ is the set of active 
    colors, 
    
    \item for each active color $i\in\dom(\chi)$, 
    $\delta(i)=(\fin(i),q(i),a(i),\fout(i))=\tile_\rho(\chi(i))$ is the current 
    $\rho$-tile at vertex $\chi(i)$ in $G$.
  \end{enumerate}
  The only initial state is the empty map $\delta_\emptyset$ with 
  $\dom(\delta_\emptyset)=\emptyset$. 
  This is also the only final state, which is reached on a \kPT
  $\tau$ if all colors have been forgotten: $\wsem{\tau}=(G,\chi_\emptyset)$.

  Transitions of the word automaton $\Bc_k$ are given 
  in Table~\ref{tbl:transitions-word-B}.  
  As above, we write $\delta(i)=(\fin(i),q(i),a(i),\fout(i))$ and
  $\delta'(i)=(\fin'(i),q'(i),a'(i),\fout'(i))$.

\begin{table}[tbp]
  \centering
  \noindent
  \begin{tabular}{|p{15mm}|p{115mm}|}
    \hline
    $\delta\xrightarrow{(i,a)}\delta'$
    &
    if $i\notin \dom(\delta)$. 
    Then, $\dom(\delta')=\dom(\delta)\cup\{i\}$,
    $\delta'(j)=\delta(j)$ for all $j\in \dom(\delta)$, and 
    $\delta'(i)=(\fempty,q,a,\fempty)$ for some $q\in Q$.
    \par The weight of this transition is $\sone$.
    \\ \hline
    $\delta\xrightarrow{\forget{i}}\delta'$ 
    &
    if $i\in \dom(\delta)$. Then $\delta'$ is the restriction of 
    $\delta$ to $\dom(\delta')=\dom(\delta)\setminus\{i\}$.
    \par The weight of this transition is $\Weight(\delta(i))$.
    \\ \hline
    $\delta\xrightarrow{\add{i,j}{\gamma}}\delta'$ 
    &
    if $i,j\in \dom(\delta)$, $i\neq j$, 
    $\gamma\notin\dom(\fout(i))$ and $\gamma\notin\dom(\fin(j))$.
    \par
    Then, $\dom(\delta')=\dom(\delta)$, $\delta'(\ell)=\delta(\ell)$ 
    for all $\ell\in \dom(\delta)\setminus\{i,j\}$,\par 
    $\delta'(i)=(\fin(i),q(i),a(i),\fout(i)\cup[\gamma\mapsto q(j)])$, \par
    $\delta'(j)=(\fin(j)\cup[\gamma\mapsto q(i)],q(j),a(j),\fout(j))$.
    \par The weight of this transition is $\sone$.
    \\ \hline
  \end{tabular}
  \caption{Transitions of the weighted word automaton $\Bc_k$.}
  \label{tbl:transitions-word-B}
\end{table}
  
  The number of partial maps from $A$ to $B$ is $(1+|B|)^{|A|}$.  Hence, the
  number of states of $\Bc_k$ is $(1+|\Tiles|)^{1+k}$.  In a tile
  $(\fin,q,a,\fout)\in\Tiles$, both $\fin$ and $\fout$ can be seen as partial
  maps from $\Gamma$ to $Q$.  Hence,
  $|\Tiles|=(1+|Q|)^{2|\Gamma|}\cdot|Q|\cdot|\Sigma|$.  Also,
  $|\Omega_k|=(1+k)(|\Sigma|+1)+(1+k)^{2}|\Gamma|$.  We deduce that, if
  $\Sigma,\Gamma,k$ are fixed, the automaton $\Bc_k$ can be constructed in
  polynomial time wrt.\ the given WTS $\wts$.
  
  Notice that we can reduce the size of $\Bc_k$ if we only
  consider states $\delta\colon[k]\to\Tiles$ such that for all
  $i\in\dom(\delta)$ the tile $\delta(i)=(\fin(i),q(i),a(i),\fout(i))$ is a
  subtile of some tile $t=(\fin',q(i),a(i),\fout')\in\Tiles$ with
  $\Weight(t)\neq\szero$.  By subtile we mean that $\fin(i)$ is the restriction
  of $\fin'$ to $\dom(\fin(i))$, i.e., $\fin(i)(\gamma)=\fin'(\gamma)$ for all
  $\gamma\in\dom(\fin(i))\subseteq\dom(\fin')$; and similarly $\fout(i)$ is the
  restriction of $\fout'$ to $\dom(\fout(i))$.
\end{proof}

\subparagraph{Evaluation of a weighted word automaton.}

\begin{lemma}\label{lem:weighted-word-automaton-evaluation}
  Given a weighted word automaton $\Bc$ and an input word 
  $w\in\Sigma^{*}$, we can compute $\sem{\Bc}(w)$ in time 
  $\mathcal{O}(|\Bc|\cdot|w|)$.
\end{lemma}

\begin{proof}
  We defined a weighted word automaton as a tuple $\Bc=(Q,T,I,F,\Weight)$.
  Another equivalent representation of $\Bc$ allows to
  compute efficiently the value $\sem{\Bc}(w)$ on a given word $w\in\Sigma^{*}$.
  Assume that $Q=\{1,\ldots,n\}$.  We view $I\in\{0,1\}^Q$ as a row vector and
  $F\in\{0,1\}^Q$ as a column vector.  For each $a\in\Sigma$, we let $\mu(a)\in
  S^{Q\times Q}$ be the $n\times n$ matrix defined by
  $\mu(a)_{i,j}=\Weight((i,a,j))\in S$ (giving weight $\szero$ for missing
  transitions, we may assume wlog that $T=Q\times\Sigma\times Q$).  Square
  matrices over the semiring $S$ form a monoid with matrix multiplication.  Hence,
  $\mu$ extends to a morphism $\mu\colon\Sigma^{*}\to S^{Q\times Q}$ by
  $\mu(w)=\mu(a_1)\cdot\mu(a_2)\cdots\mu(a_m)$ if $w=a_1a_2\cdots a_m$.
  Using distributivity of the semiring $S$, we obtain
  $\sem{\Bc}(w)=I\cdot\mu(w)\cdot F
  =I\cdot\mu(a_1)\cdot\mu(a_2)\cdot\cdots\cdot\mu(a_m)\cdot F$. Computing these 
  products from left to right (left associativity), we perform $n$ products of a 
  row vector by a matrix, and finally the product of the resulting row vector
  $I\cdot\mu(a_1)\cdot\mu(a_2)\cdot\cdots\cdot\mu(a_m)$ by the column vector $F$.
  Assuming that sum and product in the semiring $S$ take constant time, the 
  product of a row vector by a matrix takes time $\mathcal{O}(n^{2})$. Hence the 
  overall time complexity of this evaluation is $\mathcal{O}(m\cdot 
  n^{2})=\mathcal{O}(|w|\cdot|\Bc|)$.
\end{proof}

\subsection{Bounded tree-width evaluation}

We extend the efficient evaluation of WTS for graphs of bounded path-width to
graphs of bounded tree-width, which is a larger class of graphs.  For instance,
nested words may have unbounded path-width but their tree-width is at most 2.
As for path-width, tree-width can be defined via tree decompositions: instead of
a sequence of subsets of vertices, we use a tree of subsets of vertices.  Since
we will use weighted tree automata to achieve the efficient evaluation over
graphs of bounded tree-width, we define directly tree terms.  These are similar
to \kPTs, with an additional binary union $\ttunion$.

\subparagraph{Tree terms (\TTs)}  form an algebra to define labeled graphs.  With
$a \in \Sigma$, $\gamma \in \Gamma$ and $i,j\in[k]=\{0,1,\ldots,k\}$, the syntax
of \kTTs over $(\Gamma,\Sigma)$ is given by
$$
\tau ::= (i,a) \mid \add{i,j}{\gamma} \tau \mid \forget{i} \tau
\mid \tau \ttunion \tau 
$$
Each \kTT represents a colored graph $\wsem\tau=(G_\tau,\chi_\tau)$ where
$G_\tau$ is a $(\Gamma,\Sigma)$-labeled graph and $\chi_\tau\colon [k]\to V$ is
a partial injective function coloring some vertices of $G_\tau$. Colors in 
$\dom{\chi_\tau}$ are said to be active in $\tau$. 
The semantics is defined as for $\kPTs$: a leaf $(i,a)$ creates a graph with a
single $a$-labeled vertex with color $i$, $\forget{i}$ removes color $i$ from
the domain of the color map, and $\add{i,j}{\alpha}$ adds an $\alpha$-labeled
edge between the vertices colored $i$ and $j$ (if such vertices exist).  
Formally, if $\wsem{\tau}=(V,(E_\gamma)_{\gamma\in\Gamma},\lambda,\chi)$ then
\begin{itemize}[nosep]
  \item 
  $\wsem{\add{i,j}{\alpha}\tau}=(V,(E'_\gamma)_{\gamma\in\Gamma},\lambda,\chi)$ 
  with $E'_\gamma=E_\gamma$ if $\gamma\neq\alpha$ and
  $$
  E'_\alpha= 
  \begin{cases}
    E_\alpha & \text{if } \{i,j\}\not\subseteq\dom(\chi) \\
    E_\alpha\cup\{(\chi(i),\chi(j))\} & \text{otherwise.}
  \end{cases}
  $$

  \item
  $\wsem{\forget{i}\tau}=(V,(E_\gamma)_{\gamma\in\Gamma},\lambda,\chi')$ 
  with $\dom(\chi')=\dom(\chi)\setminus\{i\}$ and $\chi'(j)=\chi(j)$ for all 
  $j\in\dom(\chi')$.
  
%
\end{itemize} 
The main difference with \kPTs is $\ttunion$ which takes the union of the two graphs,
merging vertices with the same colors, if any.
\begin{itemize}[nosep]
  \item 
  Formally, consider $\tau'\ttunion\tau''$ with
  $\wsem{\tau'}=(G',\chi')=(V',(E'_\gamma)_{\gamma\in\Gamma},\lambda',\chi')$ and
  $\wsem{\tau''}=(G'',\chi'')=(V'',(E''_\gamma)_{\gamma\in\Gamma},\lambda'',\chi'')$.
  Let $I=\dom(\chi')\cap\dom(\chi'')$ be the set of colors that are defined in
  both graphs.  Wlog, we may assume that $V'\cap V''=\chi'(I)=\chi''(I)$ and
  $\chi'(i)=\chi''(i)$ for all $i\in I$, i.e., we may rename the vertices so
  that the shared colors define the shared vertices.
  The union $\tau'\ttunion\tau''$ is well-defined only if
  the shared vertices have the same labels:
  $\lambda'(\chi'(i))=\lambda''(\chi''(i))$ for all $i\in I$.
  Then, $\wsem{\tau'\ttunion\tau''}=(G'\cup G'',\chi'\cup\chi'')
  =(V,(E_\gamma)_{\gamma\in\Gamma},\lambda,\chi)$ where
  $V=V'\cup V''$, $\lambda=\lambda'\cup\lambda''$, and 
  $E_\gamma=E'_\gamma\cup E''_\gamma$ for all $\gamma\in \Gamma$.  
\end{itemize}
The tree-width of a nonempty graph $G$ is the least $k\geq1$ such that 
$G=G_\tau$ for some \kTT $\tau$.

\smallskip
Trees have tree-width 1, and as a special case, words also have tree-width 1.
Nested words have tree-width (at most) 2~\cite{MadhusudanParlato-POPL11}.  They
are words with an additional binary relation from pushes to matching pops, which
are used to represent behaviours of pushdown automata.
On the other end, grids as used for instance in Example~\ref{ex:permanent},
have unbounded tree-width.  More precisely, an $n\times n$ grid has tree-width
$n$.

\smallskip
We will focus on a regular subset of terms which ensures that the semantics is
well-defined and that the \kTTs do not contain redundant operations such as
$\add{i,j}{\gamma}\add{i,j}{\gamma}\tau$ or
$\add{i,j}{\gamma}\tau_1\ttunion\add{i,j}{\gamma}\tau_2$.
A \kTT is \emph{well-formed} if the following are satisfied:
\begin{enumerate}[nosep]
  \item  if $\forget{i}\tau'$ is a subterm of $\tau$ then $i$ is active in $\tau'$,

  \item if $\add{i,j}{\gamma}\tau'$ is a subterm of $\tau$ then $i,j$ are
  active in $\tau'$ and the edge $\gamma$ was not already added in $\tau'$
  between $\chi_{\tau'}(i)$ and $\chi_{\tau'}(j)$.

  \item if $\tau'\ttunion\tau''$ is a subterm of $\tau$ then for all $i,j$ that
  are active in both $\tau'$ and $\tau''$, the vertices $\chi_{\tau'}(i)$ and
  $\chi_{\tau''}(i)$ have the same label from $\Sigma$, and we do not already
  have a $\gamma$-edge both between $(\chi_{\tau'}(i),\chi_{\tau'}(j))$ and
  $(\chi_{\tau''}(i),\chi_{\tau''}(j))$.
\end{enumerate}

\smallskip
The problem \kTW-\feval\ is to compute $\sem{\wts}(G)$, given a WTS $\wts$ and a
$(\Gamma,\Sigma)$-graph $G$ of tree-width at most $k$.

\begin{theorem}\label{thm:kTW-feval}
  The problem \kTW-\feval\ can be solved in linear time wrt.\ the input graph
  $G$ and polynomial time wrt.\ the input WTS $\wts$.
\end{theorem}

\begin{proof}
  The proof follows the same three steps as for Theorem~\ref{thm:kPW-feval}
  using tree terms instead of \kPTs and weighted tree automata instead of
  weighted word automata.
  \begin{enumerate}[nosep]
    \item From the input graph $G$, which is assumed to be of tree-width at most
    $k$, we compute in linear time a tree decomposition using Bodlaender's
    algorithm \cite{Bodlaender_1996}.  Then, similarly to
    Lemma~\ref{lem:path-term-and-decomposition}, we compute in linear time a
    well-formed \kTT $\tau$ such that $\wsem{\tau}=(G,\emptyset)$.  In
    particular, $|\tau|=\mathcal{O}(|G|)$.
  
    \item Using Lemma~\ref{lem:WTS-2-WtreeA-k} below, from the WTS $\wts$ we
    construct in polynomial time an equivalent weighted tree automaton $\Bc_k$
    on graphs of tree-width at most $k$:
    $\sem{\wts}(G)=\sem{\Bc_k}(\tau)$.
  
    \item  We compute $\sem{\Bc_k}(\tau)$ with Algorithm~\ref{alg:tree-eval}.
    The main complexity comes from the call \textsc{TreeEval}.  Executing the
    body of this function (without the recursive calls) takes time
    $\mathcal{O}(|\Bc_k|)$.  Hence, the overall time complexity of this evaluation
    is $\mathcal{O}(|\tau|\cdot|\Bc_k|)$.
    \qedhere
  \end{enumerate}
\end{proof}

A weighted (binary) tree automaton over alphabet $\Sigma$ is usually given as a
tuple $\Bc=(Q,T,F,\Weight)$ where $F\subseteq Q$ is the subset of accepting
states, $T\subseteq(\{\bot\}\cup Q\cup Q^{2})\times\Sigma\times Q$ defines
the bottom-up transitions and $\Weight\colon T\to S$ gives weights to
transitions.  This is an equivalent representation of a
WTS over $(\{\nearrow,\nwarrow\},\Sigma)$.

\begin{lemma}\label{lem:WTS-2-WtreeA-k}
  Given a WTS $\wts$ over $(\Edgenames,\Nodelabels)$-graphs and $k>0$, we can
  compute in polynomial time wrt.\ $\wts$, a weighted tree automaton $\Bc_k$ which
  is equivalent to $\wts$ over graphs of tree-width at most $k$.
  Here, equivalent means that for all well-formed \kTTs $\tau$ with
  $\wsem{\tau}=(G,\emptyset)$, we have $\sem{\wts}(G)=\sem{\Bc_k}(\tau)$.  
\end{lemma}

\begin{proof}
  Let $\wts = (\States, \Tiles, \Weight)$ be a WTS over
  $(\Edgenames,\Nodelabels)$-graphs.  By adding tiles with weight $\szero$, we 
  may assume wlog that $\Tiles$ contains all possible tiles. Fix $k\geq1$. 
  
  \begin{table}[btp]
    \centering
    \begin{tabular}{|p{17mm}|p{110mm}|}
      \hline
      $\bot \xrightarrow{(i,a)}\delta$ 
      &
      if $\dom(\delta)=\{i\}$ and $\delta(i)=(\fempty,q,a,\fempty)$ for some $q\in Q$.
      \par The weight of this transition is $\sone$.
      \\ \hline
      $\delta\xrightarrow{\add{i,j}{\gamma}}\delta'$ 
      &
      if $i,j\in \dom(\delta)$, $i\neq j$, 
      $\gamma\notin\dom(\fout(i))$ and $\gamma\notin\dom(\fin(j))$.
      \par
      Then, $\dom(\delta')=\dom(\delta)$, $\delta'(\ell)=\delta(\ell)$ for all 
      $\ell\in\dom(\delta)\setminus\{i,j\}$,\par 
      $\delta'(i)=(\fin(i),q(i),a(i),\fout(i)\cup[\gamma\mapsto q(j)])$, \par
      $\delta'(j)=(\fin(j)\cup[\gamma\mapsto q(i)],q(j),a(j),\fout(j))$.
      \par The weight of this transition is $\sone$.
      \\ \hline
      $\delta\xrightarrow{\forget{i}}\delta'$ 
      &
      if $i\in \dom(\delta)$. Then $\delta'$ is the restriction of 
      $\delta$ to $\dom(\delta')=\dom(\delta)\setminus\{i\}$.
      \par The weight of this transition is $\Weight(\delta(i))$.
      \\ \hline
      $\delta',\delta''\xrightarrow{\ttunion}\delta$
      &
      if for all $i\in \dom(\delta')\cap \dom(\delta'')$ we have: $q'(i)=q''(i)$, $a'(i)=a''(i)$, and \par
      $\dom(\fin'(i))\cap\dom(\fin''(i))=\emptyset=\dom(\fout'(i))\cap\dom(\fout''(i))$.
      
      Then, $\delta$ is the union of $\delta'$ and $\delta''$: 
      $\dom(\delta)=\dom(\delta')\cup \dom(\delta'')$, \par
      $\delta(i)=\delta'(i)$ for all $i\in \dom(\delta')\setminus \dom(\delta'')$, \par
      $\delta(i)=\delta''(i)$ for all $i\in \dom(\delta'')\setminus \dom(\delta')$, and \par
      $\delta(i)=(\fin'(i)\cup\fin''(i),q'(i),a'(i),\fout'(i)\cup\fout''(i))$ for 
      $i\in \dom(\delta')\cap \dom(\delta'')$.
      \par The weight of this transition is $\sone$.
      \\ \hline
    \end{tabular}
    \caption{Transitions of the weighted tree automaton $\Bc_k$.}
    \label{tbl:transitions-tree-B}
  \end{table}

  A state of $\Bc_k$ is a partial map $\delta\colon[k]\to\Tiles$.
  When reading a \kTT $\tau$ with $\wsem{\tau}=(G,\chi)$, the automaton will guess a 
  labelling $\rho\colon V\to\States$ of vertices of $G$ with states of $\wts$ and 
  will reach a state $\delta$ satisfying the following two conditions:
  \begin{enumerate}[nosep]
    \item $\dom(\delta)=\dom(\chi)\subseteq[k]$ is the set of active colors, 
    
    \item for each active color $i\in\dom(\chi)$, 
    $\delta(i)=(\fin(i),q(i),a(i),\fout(i))=\tile_\rho(\chi(i))$ is the current 
    $\rho$-tile at vertex $\chi(i)$ in $G$.
  \end{enumerate}
  The only accepting state is the empty map $\delta_\emptyset$ with 
  $\dom(\delta_\emptyset)=\emptyset$, which is reached on a \kTT
  $\tau$ if all colors have been forgotten: $\wsem{\tau}=(G,\chi_\emptyset)$.
  
  The bottom-up transitions of the tree automaton $\Bc_k$ are given in
  Table~\ref{tbl:transitions-tree-B}.  
  As above, for $i\in\dom(\delta)$ we write
  $\delta(i)=(\fin(i),q(i),a(i),\fout(i))$ and similarly for $\delta'$ and
  $\delta''$.
  
  The analysis of the number of states of $\Bc_k$ is as in the proof of
  Lemma~\ref{lem:WTS-2-WWA-k}.  We deduce that, if $\Sigma,\Gamma,k$ are fixed,
  the automaton $\Bc_k$ can be constructed in polynomial time wrt.\ the given
  WTS $\wts$.
\end{proof}

\begin{algorithm}[t]
  \caption{Evaluation algorithm for a weighted tree automaton $\Bc=(Q,T,F,\Weight$).}
  \label{alg:tree-eval}
  \begin{algorithmic}[1]
    \Function{main}{$\tau\colon\mathsf{term}$}: value from $S$ \Comment Computes $\sem{\Bc}(\tau)$
      \State $\val\gets \Call{TreeEval}{\tau}$; $x\gets\szero$
      \ForAllOne{$q\in F$}{$x\gets x+\val[q]$}\EndForAllOne
      \State \Return $x$
    \EndFunction
    \Function{TreeEval}{$\tau\colon\mathsf{term}$}: array indexed by $Q$ of values from $S$ \\
      \Comment{$\textsc{TreeEval}(\tau)[q]$ is the sum of the weights of the 
      runs of $\Bc$ on $\tau$ reaching state $q$.}
      \Match{$\tau$}
        \Case{Leaf $a$}
          \ForAllOne{$q\in Q$}{$\val[q]\gets\Weight(\bot,a,q)$}\EndForAllOne
        \EndCase
        \Case{Unary $a(\tau_1)$}
          \State $\val_1\gets \Call{TreeEval}{\tau_1}$
          \ForAllOne{$q\in Q$}{$\val[q]\gets\szero$}\EndForAllOne
          \ForAllOne{$(q_1,a,q)\in T$}
            {$\val[q]\gets\val[q]+\val_1[q_1]\times\Weight(q_1,a,q)$}
          \EndForAllOne
        \EndCase
        \Case{Binary $a(\tau_1,\tau_2)$}
          \State $\val_1\gets \Call{TreeEval}{\tau_1}$; $\val_2\gets \Call{TreeEval}{\tau_2}$
          \ForAllOne{$q\in Q$}{$\val[q]\gets\szero$}\EndForAllOne
          \ForAll{$(q_1,q_2,a,q)\in T$}
            \State $\val[q]\gets\val[q]+\val_1[q_1]\times\Weight(q_1,q_2,a,q)\times\val_2[q_2]$
          \EndFor
        \EndCase
      \EndMatch
      \State \Return $\val$
    \EndFunction
  \end{algorithmic}
\end{algorithm}

\section{Discussions and conclusions}\label{sec:discussion}

\subparagraph{Connections with CSP.} The quantitative versions of the constraint
satisfaction problem (CSP) are closely related to the evaluation problem for
weighted tiling systems and graphs.  Classic (boolean) CSPs ask for the
existence of a solution of a set of constraints, as non-deterministic automata
ask for the existence of an accepting run.  In the \emph{valued}-CSP (see
e.g.~\cite{KrokhinZivny-DFU17}), weights (costs) are assigned to each constraint
depending on how the constraint is fulfilled, these weights are summed over all
constraints and the aim is to minimize this total cost.  This corresponds to our
evaluation problem in the \textsf{min-plus} tropical semiring.

The weighted counting CSP (weighted \#CSP) is defined similarly but uses a
$(+,\times)$-semiring such as $\N$, $\Z$, $\Q$, \ldots, (see
e.g.~\cite{BulatovDGJJR-jcss12,CarbonnelCooper-Constraints16}).  The cost of a
solution is the product of the weights over all constraints and the
value of the weighted \#CSP is the sum over all solutions.
Counting CSP (\#CSP) is obtained with semiring \Natural\ when functions in the
language only take values 0 or 1, thus counting the number of solutions of the
classic CSP.

One of the main problems in CSP is to determine conditions under which the
problems are tractable (polynomial time).  Feder and Vardi conjectured
\cite{FederVardi-saimcomp98} that, depending on the constraint language
$\Gamma$, problems in CSP($\Gamma$) are either in \textsf{P} or
\textsf{NP}-complete.  The dichotomy conjecture extends to \#CSP($\Gamma$),
saying that such counting problems are either in \textsf{FP} or
\#\textsf{P}-complete, see
e.g.~\cite{BulatovDalmau-IC07,DyerGoldbergJerrum-siamcomp09}.
In this paper, we show that for WTS, the evalution problem is 
\#\textsf{P}-complete (Theorem~\ref{thm:complexity}).

Most often the \emph{non-uniform} complexity is considered, meaning that the language
(for us the WTS) is not part of the input and the complexity only depends on the
instance (for us the input graph).  One such structural restriction is when the
constraint graph of the instance has bounded tree-width.  This is indeed related
to our efficient evaluation described in
Section~\ref{sec:complexity-btw-graphs}.  Our approach is different though since
we reduce WTS to weighted word/tree automata and obtain a complexity linear in
the input graph.

As future work, we plan to investigate more closely the relationship between 
weighted \#CSP and the evaluation problem for WTS. In particular, it would be 
interesting to see whether results on approximate computation which are widely 
studied for quantitative CSP can be transfered to weighted tiling systems.

\subparagraph{On the generality of the model.} 
Even though our WTS is defined to run over bounded degree graphs, we have seen
(cf.\ Remark~\ref{rem:arbitrarydegree})
that we can naturally model computational problems on arbitrary graphs that
can be input as the adjacency matrix.
The model of WGA \cite{DrosteD15} additionally has occurrence constraints
(boolean combinations of constraints of the form $\#\tile \ge n$, where $\tile
\in \Tiles$ and $n \in \N$).  A run is valid only if the occurrence constraints
are satisfied.  We could allow these constraints as well, without compromising
the complexity upper bounds.  In fact, we can allow more expressive
quantifier-free Presburger constraints on the tiles (e.g.,  $\#\tile_1 + \#\tile_2
= \#\tile_3$).  The NP machine witnessing the upper bounds can compute the
Parikh vector of the tiles used in a guessed run, and check in polynomial time
whether the constraints are satisfied.

\subparagraph{Variants.}
The evaluation problem $\fval$ is a function problem.  The decision variants
 correspond to threshold languages such as, is the computed weight $\{>,
\geq, <, \leq, =, \neq\}$ $s$, $s$ being a threshold.  There are further
variants depending on whether the threshold $s$ is part of the input or is
fixed.  The complexity depend on the semiring as well as on the value of the threshold when it is
fixed. 
 
\subparagraph{Conclusion.} 
We have given tight complexity bounds for the evaluation problem 
for various semirings.  Our complexity upper bounds allows
weights to be given in binary for problems over $(+, \times)$-semirings.
However for tropical semirings the weights are assumed to be given in unary.
While our upper bounds hold for arbitrary graphs, lower bounds are given
uniformly for pictures (grid graphs).  Further if we assume that the input graph
does not have unbounded grid as a minor (bounded tree-width), then we provide
efficient evaluation algorithm.



\bibliography{weighted-tiling-automata}

\end{document}